\documentclass[journal,10pt,twocolumn,twoside]{IEEEtran}
\usepackage{amsmath,amsfonts,amssymb,amsbsy,bm,paralist,theorem,color}
\usepackage{graphicx,relsize}
\usepackage{algorithm}
\usepackage{algpseudocode}
\usepackage{multicol}
\usepackage{multirow}
\usepackage{subcaption}
\usepackage{cite}
\usepackage{hyperref}
\usepackage{cuted}
\usepackage{pifont}
\usepackage{comment}
\usepackage{xcolor}
\usepackage{flushend}

\newtheorem{lem}{Lemma}
\newtheorem{rem}{Remark}

\graphicspath{{fig/}}

\definecolor{orange}{RGB}{255,107,0}
\definecolor{green}{RGB}{0,160,20}

\begin{document}
\title{On Dual-Fed Pinching Antenna Systems with In-Waveguide Attenuation}

\author{Ximing Xie,~\IEEEmembership{Member,~IEEE}, Hao Qin,~\IEEEmembership{Member,~IEEE}, Fang Fang,~\IEEEmembership{Senior Member,~IEEE},\\ and Xianbin Wang,~\IEEEmembership{Fellow,~IEEE}

\thanks{Ximing Xie, Fang Fang and Xianbin Wang are with the Department of Electrical and Computer Engineering, and Fang Fang is also with the Department of Computer Science, Western University, London, ON N6A 3K7, Canada (e-mail: \{xxie269, fang.fang, xianbin.wang\}@uwo.ca).}
\thanks{Hao Qin is with the School of Electrical and Electronic Engineering, University College Dublin, Dublin, D04 V1W8, Ireland (e-mail: hao.qin@ucdconnect.ie).}
\vspace{-1cm}
}\maketitle
%##################################################################
\begin{abstract}
Pinching antenna systems (PAS) have recently emerged as a promising architecture for flexible and reconfigurable wireless communications. However, their performance is fundamentally constrained by in-waveguide attenuation, which is non-negligible in practical dielectric waveguides and can severely degrade the achievable data rate, particularly for long waveguides. To overcome this limitation, we propose a dual-fed PAS (DF-PAS), in which each waveguide is equipped with two feed points located at the two ends, enabling dynamic feed-point selection based on user locations. This design effectively shortens the in-waveguide propagation distance and mitigates attenuation-induced power loss without modifying the waveguide structure or the PA actuation mechanism. We investigate the DF-PAS in both single- and multi-waveguide scenarios. For the single-waveguide case, we derive closed-form high-SNR approximations of the ergodic rate and obtain closed-form solutions for the optimal PA position and feed-point selection under time-division multiple access (TDMA). We then extend DF-PAS to a multi-waveguide scenario, where we first derive closed-form high-SNR approximations of the ergodic rate and then formulate a joint optimization problem over feed-point selection, PA placement, and beamforming under general orthogonal multiple access (OMA). To solve this problem efficiently, we develop a two-phase optimization framework that integrates greedy feed-point switching, gradient-based PA placement, and WMMSE-based beamforming. Simulation results demonstrate that the proposed DF-PAS consistently outperforms conventional single-fed PAS (SF-PAS) across various network configurations, validating its effectiveness as a practical and scalable solution for mitigating in-waveguide attenuation in PAS-enabled wireless networks.
\end{abstract}


\section{Introduction}
The evolution toward the next generation wireless networks has spurred a paradigm shift from static to dynamic antenna architectures \cite{new2024tutorial, zhu2023movable}. For example, reconfigurable intelligent surfaces (RIS) manipulate electromagnetic propagation environments via phase shifting elements to fine-tune wireless channels \cite{wu2019intelligent}. Fluid antenna systems (FAS) and movable antenna systems (MAS) utilize conductive fluids or mechanical displacement, respectively, to dynamically optimize the transceiver's position within a continuous spatial field \cite{wong2020fluid, zhu2023movable}. However, they exhibit fundamental limitations in combating large-scale fading due to line-of-sight (LoS) link blockages. To address this challenge, DOCOMO proposed the concept and developed a prototype of pinching antenna systems (PAS) in 2022 \cite{suzuki2022pinching}. A pinching antenna system typically comprises dielectric waveguides equipped with low-cost dielectric pinchers, referred to as pinching antennas (PA). Electromagnetic waves radiate from these PAs, allowing the system to establish LoS links by placing PAs at optimal positions. \par

Although PASs show strong potential for enabling reliable wireless connectivity, the technology remains at an early stage of development and faces several challenges for practical deployment. One fundamental challenge is in-waveguide attenuation, whereby the signal power gradually decays as it propagates along the waveguide. As a result, PAs located farther from the feed point radiate substantially lower power, which can severely degrade the quality of service (QoS) experienced by users. However, most existing works on PASs neglect the effect of in-waveguide attenuation and assume lossless signal propagation within the waveguide \cite{ding2025flexible, 11165763, shan2025exploiting}. This issue remained largely unexplored until the authors of \cite{tyrovolas2025performance, xu2025pinching2} investigated the effect of in-waveguide attenuation on the performance of PASs. Although in-waveguide attenuation has been recognized, it remains insufficiently investigated. Since PASs commonly employ dielectric waveguides, this issue becomes non-negligible. \par

To mitigate the effect of in-waveguide attenuation, this paper proposes a dual-fed pinching antenna system (DF-PAS). In a DF-PAS, the waveguide is equipped with two feed points located at its two ends, enabling electromagnetic waves to propagate toward the PA from two directions. Depending on user locations, the controller can flexibly select the appropriate feed point for signal injection. The effective in-waveguide propagation distance is reduced, thereby mitigating in-waveguide attenuation. Moreover, DF-PASs preserve the original waveguide architecture and PA actuation mechanism, requiring only an additional feed point. As a result, the proposed DF-PAS offers an effective mechanically simple solution that is readily compatible with existing PAS implementations.

\subsection{Related Works}
As an emerging technology, numerous works have focused on performance analysis and optimization for PASs \cite{ding2025flexible, tyrovolas2025performance, cheng2025performance, tegos2025minimum, xu2025rate, xie2025low, 11131179, 11123791}. The mathematical modeling framework of PASs was first established in \cite{ding2025flexible}, which provides the foundation for subsequent studies. Building on this model, \cite{tyrovolas2025performance} analyzed the performance of PASs while explicitly accounting for in-waveguide attenuation. The authors of \cite{cheng2025performance} further extended the performance analysis to non-orthogonal multiple access (NOMA) scenarios. From optimization perspectives, the uplink and downlink data rates of PASs were maximized in \cite{tegos2025minimum} and \cite{xu2025rate}, respectively, by optimizing the PA positions. To further reduce computational complexity, a low-complexity PA position optimization approach for sum rate maximization was proposed in \cite{xie2025low}. In addition to the data rate, energy efficiency has also been investigated as a key metric in PAS-enabled networks in \cite{11131179, 11123791}. To further explore the potential of PASs in practical environments, several works have investigated PAS-enabled systems under blockage conditions \cite{11036558, 11178241, 11315149, xie2026blcokage}. Specifically, \cite{11036558, 11178241, 11315149} adopted probabilistic blockage models, while \cite{xie2026blcokage} developed a deterministic blockage model to capture geometry-dependent blocking effects. \par

Recent works have also integrated PASs with other advanced wireless technologies, including integrated sensing and communication (ISAC) \cite{11122551, li2025pinching}, mobile edge computing (MEC) \cite{11310324, hu2025pass}, and federated learning (FL) \cite{wu2025straggler, lin2025pinching, fang2025pinching}. In the context of ISAC, \cite{11122551} maximized the communication data rate while satisfying sensing performance requirements through PA position optimization, whereas \cite{li2025pinching} minimized the Cramér–Rao bound (CRB) under QoS constraints. For PAS-enabled MEC, \cite{11310324} maximized the computational capacity by jointly optimizing device transmit power, time allocation, and PA positions, while \cite{hu2025pass} improved task offloading efficiency and latency performance by exploiting PAS-assisted transmission. In the context of FL, \cite{wu2025straggler, lin2025pinching} employed PASs to mitigate the straggler problem, and \cite{fang2025pinching} investigated multiple FL scenarios in which PASs can enhance learning efficiency. \par

However, most of the aforementioned works overlook the effect of in-waveguide attenuation. Although this issue has been recognized in a limited number of studies, an efficient and practical solution is still lacking. In particular, \cite{xu2025pinching} acknowledged the impact of in-waveguide attenuation and showed that its effect can be negligible under certain conditions. However, this conclusion relies on the assumption of an ultra-low-attenuation waveguide, which requires advanced materials. How to effectively address in-waveguide attenuation in commonly used dielectric waveguides remains an open problem. In addition, a new PAS architecture, referred to as the segmented waveguide-enabled PAS, was proposed in \cite{ouyang2025uplink}. In this design, the waveguide is divided into multiple short dielectric segments, each equipped with a dedicated PA. While this architecture effectively reduces the in-waveguide propagation distance, it requires a larger number of RF chains and introduces additional hardware complexity. Moreover, the segmented structure imposes extra manufacturing challenges due to its increased structural and implementation complexity.

\subsection{Motivation and Contributions}
Although PASs offer a promising solution to combat large-scale fading, their performance is fundamentally constrained by in-waveguide attenuation, where signal power decays exponentially with the propagation distance inside the waveguide. This effect is non-negligible for commonly used dielectric materials and can significantly degrade the achievable data rate, particularly for users located far from the feed point. Most existing works either neglect this attenuation or rely on ultra-low-loss waveguides or segmented architectures, which impose stringent material requirements or incur high hardware and implementation complexity. These limitations motivate a simple and practical solution that can reduce the effective in-waveguide propagation distance without altering the basic PAS structure. This paper proposes a DF-PAS, where each waveguide is equipped with two feed points located at its two ends. By dynamically selecting the feed point according to user location, the DF-PAS shortens the effective in-waveguide propagation distance and thereby mitigates attenuation-induced power loss. This design preserves the original waveguide structure and PA operation mechanism, requiring only an additional feed point and simple switching control. As a result, the DF-PAS offers an effective and hardware-efficient means to improve PAS performance under practical in-waveguide attenuation. The main contributions of this paper are summarized as follows:
\begin{itemize}
    \item We demonstrate that in-waveguide attenuation is non-negligible in commonly used dielectric waveguides for PASs by characterizing its dependence on material loss tangent, effective refractive index, and operating frequency. This analysis reveals substantial power decay over practical waveguide lengths at $28$ GHz. Motivated by this observation, we propose a DF-PAS in which each waveguide is equipped with two feed points at its ends, enabling feed-point selection based on user location to shorten the effective in-waveguide propagation distance without changing the waveguide structure.
    \item For the single-waveguide DF-PAS, we conduct a performance analysis that explicitly accounts for in-waveguide attenuation and derive closed-form approximations of the ergodic rate under high-SNR conditions. The results show that the DF-PAS achieves a linear rate gain with respect to both the waveguide length and the attenuation coefficient compared with the single-fed PAS (SF-PAS). We further derive closed-form expressions for the optimal PA position and establish feed-point selection conditions for the time-division multiple access (TDMA) case.
    \item  We further extend the DF-PAS to multi-waveguide scenarios and conduct a comprehensive performance analysis. Then, we formulate a joint optimization problem over beamforming, PA positions, and feed-point selection for a more general orthogonal multiple access (OMA) case. To efficiently solve the resulting problem, a two-phase optimization framework is proposed by integrating greedy feed point switching, gradient-based PA placement, and WMMSE-based beamforming.
    \item Extensive simulations are provided for both single- and multi-waveguide scenarios to validate the proposed DF-PAS architecture. The results show that the DF-PAS consistently outperforms the SF-PAS and other benchmark schemes under different waveguide lengths, transmit power levels, and network configurations. These results demonstrate the effectiveness of the DF-PAS in mitigating in-waveguide attenuation and maintaining robust performance against free-space path loss.
\end{itemize}

\subsection{Organization}
The remainder of this paper is organized as follows. Section II introduces the DF-PAS and presents the system model and performance analysis for the single-waveguide case. Section III extends the proposed DF-PAS to multi-waveguide scenarios and formulates the corresponding joint optimization problem. Section IV provides numerical results to evaluate the performance of the proposed schemes. Finally, Section V concludes the paper.

\section{Dual-Fed Pinching Antenna System: Single Waveguide Case}
\subsection{In-waveguide Attenuation}
In-waveguide attenuation is the attenuation of electromagnetic energy during propagation within a waveguide. It depends on the waveguide geometry and material properties and constitutes an inherent and unavoidable physical phenomenon. The in-waveguide attenuation is commonly modeled as
\begin{equation}
P(z) = P_{\rm in} e^{-\alpha z}, \label{attenuation model}
\end{equation}
where $\alpha$ denotes the power attenuation coefficient of the waveguide, $P_{\rm in}$ is the input power injected at the feed point, and $z$ represents the propagation distance along the waveguide measured from the feed point. Eq. \eqref{attenuation model} illustrates that a larger value of $\alpha$ corresponds to more severe in-waveguide attenuation. \par
\begin{figure}[t]
     \centering
     \includegraphics[width=0.45\textwidth]{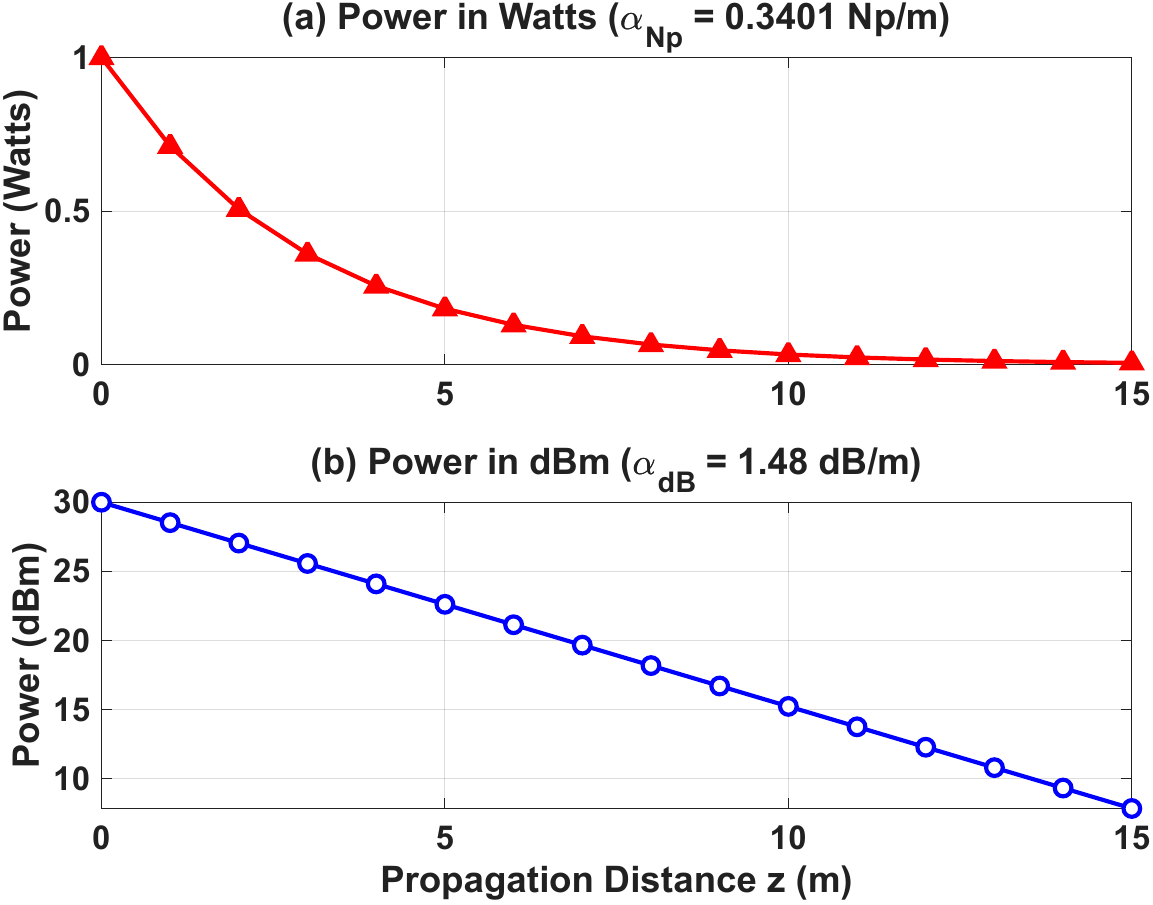}
     \caption{Power decay along a PTFE waveguide}
     \label{power decay}
\end{figure}

The pinching-antenna system employs dielectric waveguides to guide EM waves. Consequently, dielectric loss constitutes the dominant in-waveguide attenuation mechanism during electromagnetic propagation. By considering only dielectric loss, the power attenuation coefficient of a dielectric waveguide can be approximated as
\begin{equation}
\alpha_{d} \approx \frac{2 \pi \cdot n_{\rm eff} \tan(\delta)}{\lambda_0}, \label{dielectric attenuation coefficient}
\end{equation}
where $n_{\rm eff}$ denotes the effective refractive index of the guided mode, $\tan(\delta)$ is the dielectric loss tangent of the waveguide material, and $\lambda_0$ represents the free-space wavelength, given by $\lambda_0 = c/f_c$, with $c$ denoting the speed of light and $f_c$  denoting the carrier frequency. The most common real-world example of dielectric waveguides is a PTFE (Teflon) rod waveguide, whose loss tangent and refractive index are $0.0004$ and $1.45$, respectively. If we assume the frequency of EM wave is $28$ GHz, the power attenuation coefficient of PTFE is $\alpha_d = 1.48$ dB/m. Fig. \ref{power decay} illustrates the power attenuation in a PTFE waveguide, where the transmitted power decreases from $1$ W ($30$ dBm) at the input to $0.033$ W ($15$ dBm) after $10$ m of propagation. Consequently, users located far from the feed point experience severe in-waveguide attenuation, which may prevent their QoS requirements from being satisfied. These results indicate that in-waveguide attenuation is non-negligible in dielectric-waveguide-enabled pinching antenna systems.

\begin{figure}[t]
     \centering
     \includegraphics[width=0.5\textwidth]{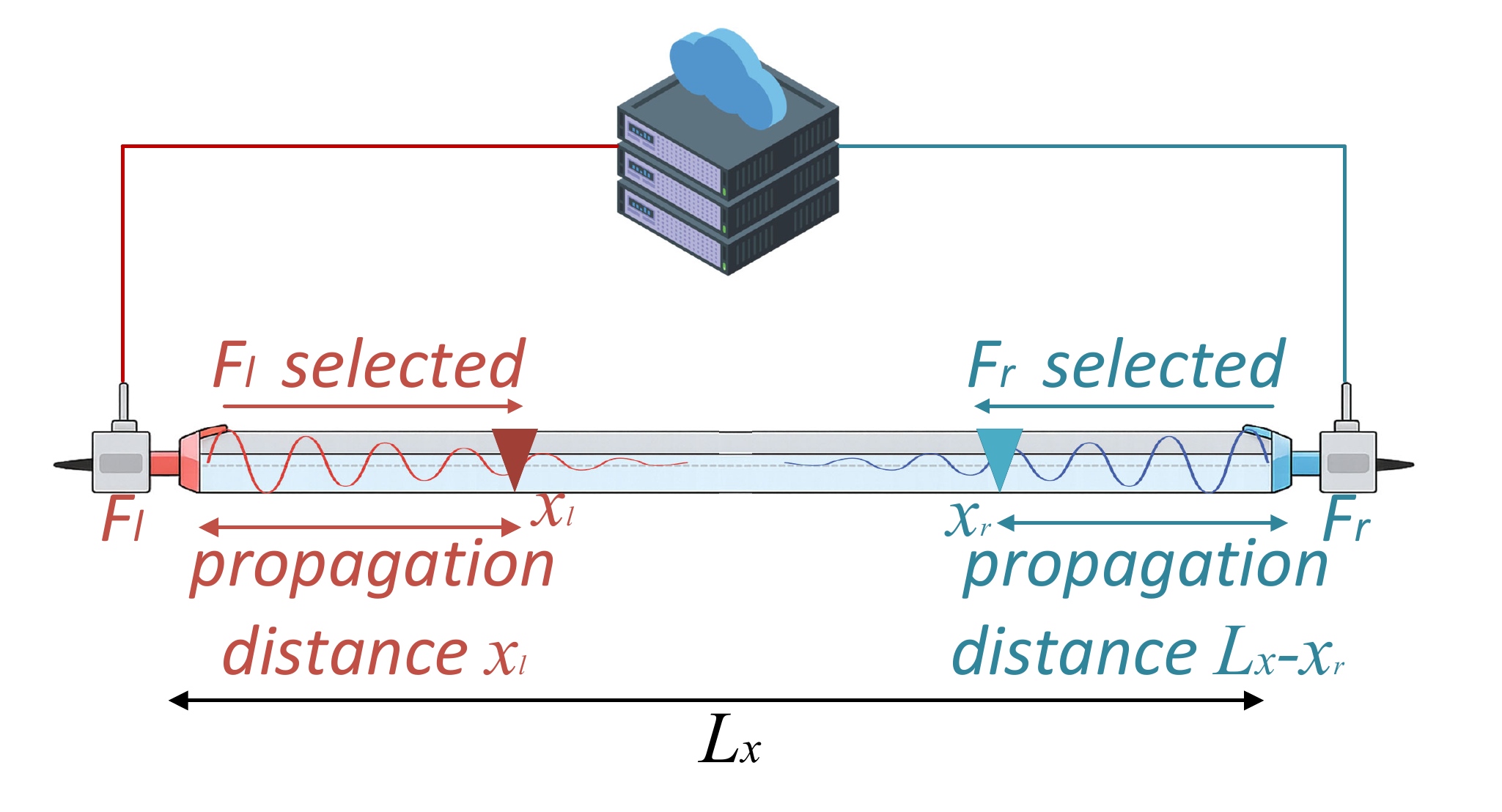}
     \caption{Architecture of the DF-PAS}
     \label{dual-fed}
\end{figure}

\subsection{Dual-Fed Pinching Antenna System}
The conventional waveguide usually has two ends and the EM waveguide is fed into it from one of the ends. As can be observed from Fig. \ref{power decay}, the user far from the feed point experiences severe in-waveguide attenuation. To mitigate this effect, we introduce a DF-PAS architecture shown in Fig. \ref{dual-fed}. The waveguide in a DF-PAS is equipped with two feed points at two ends of the waveguide, which are connected to the base station. The base station can dynamically select the feed point from either end of the waveguide according to the user location, thereby shortening the effective propagation distance within the waveguide. For example, users positioned closer to the left (right) end of the waveguide are served by injecting the EM signal from the left (right) feed point, which significantly reduces the in-waveguide propagation loss and improves the achievable received power. By enabling feed-point selection, the DF-PAS provides an additional degree of freedom for combating in-waveguide attenuation. \par

We adopt TDMA as a representative OMA scheme, where each user is assigned an orthogonal time slot to ensure interference-free transmission. Following \cite{ding2025analytical}, we assume that the PA location remains fixed within each transmission frame. However, physically repositioning the PA between two time slots would introduce excessive mechanical complexity and latency. Instead, an antenna activation mechanism \cite{wang2025antenna} can be used to rapidly switch/activate the desired PA on the waveguide, enabling low-latency reconfiguration (e.g., between frames) without continuous mechanical movement.

We consider a downlink system model where a single dual-fed dielectric waveguide equipped with one PA serves $M$ single-antenna users, denoted by ${\rm U}_m, 1 \leq m \leq M$. To accurately describe the geometric relationships in the pinching-antenna systems, we build a 3D Cartesian coordinate system. We assume that $M$ users are uniformly distributed over a rectangular service area of size $L_x\times L_y$, with locations denoted by $\psi_m = (x_m, y_m,0), 1 \leq m \leq M$. The waveguide of length $L_x$ is deployed parallel to the $x$-axis and located at the edge of the service area (i.e., at $y=0$) along the $y$-dimension at height $d$. The two ends of the waveguide serve as feed points. In each transmission interval, the signal is injected from either the left or the right feed point, but not both simultaneously. The left and right feed points are denoted by $F_l$ and $F_r$, with corresponding locations $\psi_l = (0, 0,d)$ and $\psi_r= (L_x, 0,d)$, respectively. Moreover, only one PA is deployed on the waveguide, whose location is denoted by $\psi^{\mathrm{Pin}} = (x^{\mathrm{Pin}}, 0, d)$. \par

Following the channel model in \cite{xu2025pinching3}, the channel from the PA to the user is modeled as the combination of an LoS component and an NLoS component. Specifically, a hybrid model that accounts for near-field spherical wave LoS propagation and cluster-based statistical fading in NLoS conditions is adopted. According to \cite{ding2025flexible}, the LoS channel from the PA to ${\rm U}_m$ is modeled as
\begin{equation}
    h^{\rm LoS}_m = \frac{\sqrt{\eta} e^{-j \frac{2 \pi}{\lambda} ||\psi^{\mathrm{Pin}} - \psi_m||}}{ ||\psi^{\mathrm{Pin}} - \psi_m||}, \label{LoS channel}
\end{equation}
where $\eta = \frac{c^2}{(4 \pi f_c)^2}$ is an LoS constant. In addition to the LoS component, the channel between the PA and ${\rm U}_m$ also includes an NLoS component due to multi-path scattering in the propagation environment, which can be modeled as
\begin{equation}
    h_m^{\rm NLoS} = \sum\limits_{k=1}^{N_c} g_{m,k} (x^{\mathrm{Pin}}), \label{NLoS channel}
\end{equation}
where $N_c$ denotes the number of scattering paths and $g_{m,k}$ denotes complex small-scale fading gain of the $k-$th path. $g_{m,k}$ can be modeled as an independent zero-mean circularly symmetric complex Gaussian random variable with distance-dependent variance \cite{xu2025pinching3}:
\begin{equation}
    g_{m,k} (x^{\mathrm{Pin}}) \sim \mathcal{CN} \left(0, \frac{\mu_{m,k}^2}{r_m^{2} \left(x^{\rm Pin} \right)} \right), \label{NLoS path gain}
\end{equation}
where $\mu_{m,k}^2$ represents the average power weight of the $k-$th scattering path and $r_m = ||\psi^{\mathrm{Pin}} - \psi_m||$ denotes the Euclidean distance between the PA and ${\rm U}_m$. Accordingly, the channel between the PA and ${\rm U}_m$ can be expressed as:
\begin{equation}
    h_m = h_m^{\rm LoS} + h_m^{\rm NLoS}. \label{channel model}
\end{equation} \par

Since the signal can only be injected from one feed point during each transmission, we introduce a binary variable $\xi \in \{0,1\}$ to indicate the feed-point selection. Specifically, $\xi = 1$ means the left feed point $F_l$ is active, while $\xi = 0$ means the right feed point $F_r$ is active. Let $s_m$ be the transmitted symbol for ${\rm U}_m$ with $\mathbb{E}\{|s_m|^2\} = 1$ and $P_0$ be the injected power. The signal sent from the PA can be expressed as follows:
\begin{align}
    \tilde{s}_m &= \xi \sqrt{P_0} e^{- \left(\frac{\alpha}{2} + j\frac{2 \pi }{\lambda_g}\right) ||\psi^{\mathrm{Pin}} - \psi_l|| }s_m \notag \\
    &+ (1-\xi) \sqrt{P_0} e^{- \left(\frac{\alpha}{2} + j\frac{2 \pi }{\lambda_g}\right) ||\psi^{\mathrm{Pin}} - \psi_r|| }s_m, \label{signal from PA}
\end{align}
where $\lambda_g = \frac{\lambda}{n_{\rm eff}}$ denotes the waveguide wavelength. $\lambda$ denotes the free-space wavelength. Based on the channel model \eqref{channel model}, the received signal of ${\rm U}_m$ during its allocated time slot can be expressed as:
\begin{equation}
    y_m = h_m \tilde{s}_m + n_m, \label{received signal}
\end{equation}
where $n_m \sim \mathcal{CN} (0, \sigma^2)$ denotes the additive white Gaussian noise (AWGN) with power $\sigma^2$.

\subsection{Performance Analysis}
In this subsection, we analyze the performance of the DF-PAS. The performance metric we adopted in this paper is the data rate. Since the PA is always activated at the closest point to the user, the PA–user distance is small. For short links, the LoS path is usually much stronger than scattered paths, i.e., $|h_m^{\rm LoS}|^2 \gg |h_m^{\rm NLoS}|^2$. Therefore, we ignore the NLoS component to obtain a tractable ergodic rate expression that captures the key effect of in-waveguide attenuation and distance-dependent LoS path loss. \par

Since the PA is placed closest to the user, the signal should be injected from the end of the waveguide closest to the PA. Specifically, if the PA is placed on in the left half, $F_l$ is active and the propagation distance is $z = x_m$. If the PA is placed in the right half, $F_r$ is active and the propagation distance is $z = L_x - x_m$. Therefore, the effective in-waveguide propagation distance $z$ is given by
\begin{equation}
    z = \min \left(x_m, L_x - x_m \right). \label{propagation distance}
\end{equation}
Then, \eqref{signal from PA} can be recast into
\begin{equation}
    \tilde{s}_m = \sqrt{P_0} e^{- \left(\frac{\alpha}{2} + j\frac{2 \pi }{\lambda_g}\right) \min \left(x_m, L_x - x_m \right)} s_m. \label{recast signal from PA}
\end{equation}
The received signal of ${\rm U}_m$ without NLoS component can be expressed as
\begin{equation}
    y_m = \sqrt{\frac{P_0\eta}{y_m^2 + d^2}}e^{\phi (x_m)} s_m + n_m, \label{recast received signal}
\end{equation}
where $\phi (x_m) = -\frac{\alpha}{2}  \min \left(x_m, L_x - x_m \right) -2\pi j \left(\frac{1}{\lambda}\sqrt{y_m^2 + d^2} + \frac{1}{\lambda_d} \min \left(x_m, L_x - x_m \right) \right) $ represents the distance-dependent phase. We note the $r_m = \sqrt{y_m^2 + d^2}$ in \eqref{recast received signal} due to $x^{\rm Pin} = x_m$. Therefore, the received signal-to-noise ratio (SNR) of ${\rm U}_m$ can be expressed as
\begin{equation}
    \gamma_m = \frac{P_0 \eta}{\sigma^2} \cdot \frac{e^{-\alpha  \min \left(x_m, L_x - x_m \right)} }{y_m^2 + d^2}. \label{SNR}
\end{equation}
The ergodic rate is defined as the expectation of the data rate, which is given by
\begin{equation}
    \bar{R}^{\rm DF}_m = \mathbb{E}_{x,y} \left[\log_2 (1+ \gamma_m) \right]. \label{e rate}
\end{equation}
To obtain a tractable expression, we use the high-SNR approximation, where $\log_2 (1 + \gamma_m) \approx \log_2(\gamma_m)$ when $\gamma_m \gg 1$. Then, the ergodic rate can be approximated as
\begin{align}
    &\bar{R}^{\rm DF}_m \approx \mathbb{E}_{x,y} \left[\log_2 \left(\frac{P_0 \eta}{\sigma^2} \cdot \frac{e^{-\alpha  \min \left(x_m, L_x - x_m \right)} }{y_m^2 + d^2} \right) \right] \\
    & = \int_0^{L_x} \int_0^{L_y} \log_2 \left(\frac{P_0 \eta}{\sigma^2} \cdot \frac{e^{-\alpha  \min \left(x, L_x - x \right)} }{y^2 + d^2} \right) \frac{1}{L_x} \frac{1}{L_y} \,dx\,dy. \label{int e rate}
\end{align}
\begin{lem} \label{lemma 1}
    At high SNR, the ergodic rate of ${\rm U}_m$ can be approximated as
    \begin{align}
        \bar{R}^{\rm DF}_m &\approx \log_2 \left(\frac{P_0 \eta}{\sigma^2} \right) - \frac{\alpha L_x}{4} \log_2(e) \notag \\
         &-\left[\log_2 \left(L_y^2 + d^2 \right) - \frac{2}{\ln2} \left(1-\frac{d}{L_y} \arctan \left(\frac{L_y}{d}\right) \right) \right]. \label{closed form e rate}
    \end{align}
\end{lem}
\begin{proof}
    See Appendix \ref{appendox A}
\end{proof}

\begin{lem} \label{lemma 2}
    At high SNR, the ergodic rate gain of ${\rm U}_m$ achieved by the DF-PAS compared with the SF-PAS is 
    \begin{equation}
        \Delta \bar{R} = \frac{\alpha L_x}{4} \log_2(e). \label{e rate gain}
    \end{equation}
\end{lem}
\begin{proof}
    See Appendix \ref{appendox B}
\end{proof}
\begin{rem}
    The performance advantage of the DF-PAS scales linearly with both the waveguide length $L_x$ and the attenuation coefficient $\alpha$. This indicates that the DF-PAS is particularly beneficial for scenarios involving long-distance transmission or high-loss waveguide materials.
\end{rem}
Since TDMA is adopted, the transmission frame is equally allocated to each user. Therefore, the ergodic sum rate of $M$ users can be expressed as
\begin{equation}
    \bar{R}^{\rm DF} _{\rm sum} = \frac{1}{M} \sum_{m=1}^M \bar{R}^{\rm DF}_m. \label{ergodic sum rate}
\end{equation}
\vspace{-1cm}
\subsection{Optimization}
In this subsection, we jointly optimize the feed-point selection and the PA position by considering both LoS and NLoS components. Note that placing the PA closest to the user may not be optimal, because a longer in-waveguide path can cause strong attenuation. Therefore, the PA position should balance the in-waveguide loss and the free-space path loss. \par
\subsubsection{Optimal PA Position}
Since the cases with signal injection from $F_l$ or $F_r$ are symmetric, it is enough to solve one case. Once the optimal PA position for one feed point is found, the optimal position for the other feed point is its mirror point with respect to the waveguide center. Without loss of generality, we assume the signal is injected from $F_l$, which means $\xi = 1$. In this case, the problem is to find the optimal PA position under in-waveguide attenuation, while considering both LoS and NLoS components. \par
When $\xi = 1$, the received signal of ${\rm U}_m$ is
\begin{equation}
    y_m = \sqrt{P_0} e^{- \left(\frac{\alpha}{2} + j\frac{2 \pi }{\lambda_g}\right) x^{\rm Pin}} h_m s_m + n_m. \label{ym SISO op}
\end{equation}
The derivation of the SNR of the received signal is given by
\begin{align}
    \gamma_m &= \frac{P_0 e^{-\alpha x^{\rm Pin}} \mathbb{E} \{|h_m|^2\}}{\sigma^2} \notag \\
    &=\frac{P_0}{\sigma^2} \left(\frac{\eta\;e^{-\alpha x^{\rm Pin}}}{(x^{\rm Pin} - x_m)^2 + D_m}  + e^{-\alpha x^{\rm Pin}} \mathbb{E} \{|h^{\rm NLoS}_m|^2\}\right) \notag \\
    & \overset{(a)}{=} \frac{P_0}{\sigma^2} \left(\frac{\eta\;e^{-\alpha x^{\rm Pin}}}{(x^{\rm Pin} - x_m)^2 + D_m}  + \frac{\mu_m\;e^{-\alpha x^{\rm Pin}}}{(x^{\rm Pin} - x_m)^2 + D_m}\right) \notag \\
    & = \frac{P_0 \eta_m}{\sigma^2} \cdot \frac{e^{-\alpha x^{\rm Pin}}}{(x^{\rm Pin} - x_m)^2 + D_m}, \label{snr SISO op}
\end{align}
where $D_m = y_m^2 + d^2$ denotes a constant of PA-User distance, $u_m = \sum_{k=1}^{N_c} \mu_{m,k}^2$ aggregates the effective NLoS powers and $\eta_m = \eta + \mu_m$ denotes the effective channel coefficient of ${\rm U}_m$. $(a)$ is obtained by the assumption \eqref{NLoS path gain}. We note that maximizing the data rate of ${\rm U}_m$ is equivalent to maximizing $\gamma_m$. Hence, the optimization problem can be formulated as
\begin{subequations}\label{Prob1}
    \begin{align}
        {\rm P1}: &\max_{\{x^{\rm Pin}\}}\;\; \gamma_m \label{P10}\\
		&~\mathrm{s.t.}\;\;\;\; 0 \leq x^{\rm Pin} \leq L_x. \label{P11}
    \end{align}
\end{subequations}
We note that taking the NLoS channel component into consideration does not change the form of $\gamma_m$. Therefore, the optimal solution of ${\rm P1}$ can be solved by the method proposed in \cite{xu2025pinching2}. To save space, we omit the details here. Let $x_m^{l *}$ be the optimal PA position for ${\rm U}_m$ when $F_l$ is selected. $x_m^{l *}$ is given by
\begin{equation}
     x_m^{l *} = \begin{cases}
        0, & \mbox{if}~~D_m \geq \frac{x_m(2 - \alpha x_m)}{\alpha} \\
        x_m +\frac{-1 + \sqrt{1-\alpha^2 D_m}}{\alpha}, & \textit{otherwise}
    \end{cases}. \label{optimal PA position of left}
\end{equation}
According to the symmetry, when $F_r$ is selected, the optimal PA position for ${\rm U}_m$ is given by
\begin{equation}
     x_m^{r *} = \begin{cases}
        L_x, & \mbox{if}~~D_m \geq \frac{\Tilde{x}_m(2 - \alpha \Tilde{x}_m)}{\alpha} \\
        x_m -\frac{-1 + \sqrt{1-\alpha^2 D_m}}{\alpha}, & \textit{otherwise}
    \end{cases}, \label{optimal PA position of left}
\end{equation}
where $\Tilde{x}_m = L_x - x_m$.
\subsubsection{Feed-Point Selection}
There are two options for feed-point selection. First, the feed point can be adjusted for each user. This can choose the better feed point based on the user location, but it requires switching the feed point between time slots and increases hardware complexity. Second, the feed point can be fixed for all users. This is more practical, since once the feed point is selected, it stays unchanged during the whole transmission. \par

We start with the first option where each user can select its own feed point. By considering both feed points $F_l$ and $F_r$, the SNR of ${\rm U}_m$ can be expressed as
\begin{equation}
    \gamma_m(\xi) = \frac{P_0 \eta_m}{\sigma^2} \cdot \frac{e^{-\alpha f(\xi)}}{(x^{\rm Pin} - x_m)^2 + D_m}, \label{feed point snr }
\end{equation}
where $f(\xi) = \xi x^{\rm Pin} + (1-\xi)(L_x - x^{\rm Pin})$ defined as a $\xi$-dependent function. We note that $\gamma_m(\xi)$ decreases as $f(\xi)$ increases. Therefore, maximizing $\gamma_m(\xi)$ with respect to $\xi$ is equivalent to minimizing $f(\xi)$. Let $\xi_m^*$ denote the optimal feed point indicator for ${\rm U_m}$, which is given by
\begin{align}
    \xi_m^* &= \arg \min\limits_{\xi \in \{0,1\}} f(\xi) \notag \\
    &=\begin{cases}
1, & x^{\rm Pin}\le \dfrac{L_x}{2},  \\
0, & x^{\rm Pin}>\dfrac{L_x}{2}.
\end{cases}
\end{align}
when $\xi^*_m = 1$, the left feed point $F_l$ is selected for ${\rm U}_m$; otherwise, the right feed point $F_r$ is selected. \par
For the second option, the feed point is selected before the transmission based on all users' locations and does not change during the transmission. 
Let $R_{\rm sum}^l$ and $R_{\rm sum}^r$ denote the sum rate of all users when $F_l$ and $F_r$ are selected, respectively. $R_{\rm sum}^l$ and $R_{\rm sum}^r$ can be expressed as
\begin{equation}
    R_{\rm sum}^l = \frac{1}{M}\sum\limits_{m=1}^M \log_2 \left(1 + \frac{P_0 \eta_m}{\sigma^2} \frac{e^{-\alpha x^{l*}_m}}{\left(x^{l*}_m - x_m \right)^2 + D_m} \right), \label{sum rate of left point}
\end{equation}
and
\begin{equation}
    R_{\rm sum}^r = \frac{1}{M}\sum\limits_{m=1}^M \log_2 \left(1 + \frac{P_0 \eta_m}{\sigma^2} \frac{e^{-\alpha (L_x - x^{r*}_m)}}{\left(x^{r*}_m - x_m \right)^2 + D_m} \right). \label{sum rate of right point}
\end{equation}
The optimal feed point indicator is given by
\begin{equation}
    \xi^*  =\begin{cases}
1, & R_{\rm sum}^l \geq R_{\rm sum}^r,  \\
0, & R_{\rm sum}^l < R_{\rm sum}^r.
\end{cases}
\end{equation}
\section{Dual-Fed Pinching Antenna System: Multiple Waveguides Case}
In this section, we extend the DF-PAS to the multi-waveguide scenario. Unlike the single-waveguide setup, multiple parallel waveguides are deployed to cover a wider service area, and each user is jointly served by all waveguides.
\subsection{Multi-waveguide Dual-Fed Pinching Antenna System}
We consider $N$ parallel waveguides deployed along the $x$-axis and equally spaced over the service area of width $L_y$. The $y$- coordinate of the $n$-th waveguide is denoted by $y_n$ and is given by $y_n = \frac{2n-1}{2N} L_y, n=1,...,N$. Each waveguide is equipped with one PA. The location of the PA on the $n$-th waveguide is denoted by $\psi_n^{\rm Pin} = (x^{\rm Pin}_n, y_n,d)$. Moreover, each waveguide has two feed points at its two ends, denoted by $F_{l,n}$ and $F_{r,n}$, with corresponding locations $\psi_n^l = (0,y_n,d)$ and  $\psi_n^r = (L_x,y_n,d)$. The injected power for each waveguide is denoted by $P_n$. One feasible approach in this case is to serve the users individually. Specifically, as in Section II, we can employ TDMA, where each user is assigned a dedicated time slot. However, this scheme incurs a performance loss because it does not fully exploit the spatial degrees of freedom enabled by multiple pinching antennas. In this subsection, we therefore consider a scheme in which all users are simultaneously served by the multiple pinching antennas.\par

The channel between the PA on the $n$-th waveguide and ${\rm U}_m$ can be expressed as
\begin{equation}
    h_{n,m} = h_{n,m}^{\rm LoS} + h_{n,m}^{\rm NLoS},
\end{equation}
where
\begin{equation}
    h^{\rm LoS}_{n,m} = \frac{\sqrt{\eta} e^{-j \frac{2 \pi}{\lambda} ||\psi^{\mathrm{Pin}}_n - \psi_m||}}{ ||\psi^{\mathrm{Pin}}_n - \psi_m||}. \label{LoS channel multi-waveguide}
\end{equation}
and
\begin{equation}
      h_{n,m}^{\rm NLoS} = \sum\limits_{k=1}^{N_{n,m}}\Gamma_k \frac{e^{-j \frac{2\pi}{\lambda} d_{n,k}^{1}(x_n^{\rm Pin})}}{d_{n,k}^{1}(x_n^{\rm Pin})} \cdot \frac{e^{-j \frac{2\pi}{\lambda} d_{m,k}^{2}}}{d_{m,k}^{2}}, \label{NLoS channel multi-waveguide}
\end{equation}
where $N_{n,m}$ denotes the number of scattering paths between the PA on the $n-$th waveguide and ${\rm U}_m$, $\Gamma_k$ denotes the reflective coefficient of the $k-$th scatter, $d_{n,k}^1$ denotes the distance between the PA on the $n-$th waveguide and the $k-$th scatter, and $d_{m,k}^2$ denotes the distance between ${\rm U}_m$ and the $k-$th scatter. Let $\psi_k^{\rm SC} = (x_k^{\rm SC}, y_k^{\rm SC}, 0)$ denote the position of the $k-$th scatter. Then, $d_{n,k}^1$ and $d_{m,k}^2$ are given by $d_{n,k}^1 = |\psi_n^{\rm Pin} - \psi_k^{\rm SC}|$ and $d_{m,k}^2 = |\psi_m - \psi_k^{\rm SC}|$, respectively. \par

Note that all signals injected into the same waveguide must be identical. Therefore, the superposition of all users' signals is fed into each waveguide. The signal fed into the $n$-th waveguide is given by
\begin{equation}
    s_n = \sum_{m=1}^M p_{n,m} s_m, \label{superposition signal}
\end{equation}
where $p_{n,m}$ denotes the beamforming coefficient assigned to ${\rm U}_m$ on the $n$-th waveguide, and $s_m$ denotes the information symbol intended for ${\rm U}_m$. After considering in-waveguide attenuation and feed-point selection, the signal transmitted from the PA on the $n-$th waveguide can be expressed as
\begin{align}
    x_n &= \xi_n  e^{- \left(\frac{\alpha}{2} + j\frac{2 \pi }{\lambda_g}\right) ||\psi^{\mathrm{Pin}}_n - \psi_n^l|| }s_n \notag \\
    &+ (1-\xi_n) e^{- \left(\frac{\alpha}{2} + j\frac{2 \pi }{\lambda_g}\right) ||\psi^{\mathrm{Pin}}_n - \psi_n^r|| }s_n, \label{signal from PA on n waveguide}
\end{align}
where $\xi_n$ denotes the feed-point selection indicator of the $n$-th waveguide. Let $\nu_n = \xi_n e^{-(\frac{\alpha}{2} + j \frac{2\pi}{\lambda_g}) \|\psi_n^{\rm Pin} - \psi_n^l\|} + (1 - \xi_n) e^{-(\frac{\alpha}{2} + j \frac{2\pi}{\lambda_g}) \|\psi_n^{\rm Pin} - \psi_n^r\|}$ denote the equivalent in-waveguide response. The effective composite channel between the input of the $n$-th waveguide and ${\rm U}_m$ can be expressed as $\tilde{h}_{n,m} = h_{n,m} \nu_n$. Since ${\rm U}_m$ receives signals from all waveguides, the received signal at ${\rm U}_m$ can be expressed as
\begin{equation}
    y_m = \sum_{n=1}^N \tilde{h}_{n,m} p_{n,m} s_m + \sum_{i \neq m} \sum_{n=1}^N \tilde{h}_{n,m} p_{n,i} s_i + n_m. \label{eq:received_signal}
\end{equation}
We assume that the symbols satisfy $\mathbb{E}[|s_m|^2] = 1, \forall m$. The signal-to-interference-plus-noise ratio (SINR) is expressed as
\begin{equation}
    \gamma_m = \frac{|\tilde{\mathbf{h}}_m^H \mathbf{p}_m|^2}{\sum_{i \neq m} |\tilde{\mathbf{h}}_m^H \mathbf{p}_i|^2 + \sigma^2 }, \label{SINR for user m}
\end{equation}
where $\tilde{\mathbf{h}}_m = [\tilde{h}_{1,m}, \cdots, \tilde{h}_{N,m}]^T$ is the effective channel vector, and $\mathbf{p}_m = [p_{1,m}, \cdots, p_{N,m}]^T$ is the beamforming vector.

\subsection{Performance Analysis}
In this subsection, we analyze the performance of the DF-PAS with multiple waveguides. To simplify the analysis, we consider the single user case (${\rm U}_m$) and the LoS only channel. The effective composite channel $\tilde{h}_{n,m}$ becomes the LoS channel $h_{n,m}^{\rm LoS}$, given by \eqref{LoS channel multi-waveguide}. We further assume each PA is placed closest to ${\rm U}_m$; therefore, all PAs share the same effective in-waveguide propagation distance, which is given by \eqref{propagation distance}. We assume each waveguide has the equal injected power, which means $P_n = \frac{P_0}{N}, \forall n$. The beamforming adopts maximum ratio transmission (MRT). Hence, \eqref{SINR for user m} can be recast into
\begin{equation}
\gamma_m=\frac{P_0\eta}{N\sigma^2} e^{-\alpha z}\left(\sum_{n=1}^{N}\frac{1}{\sqrt{(y_m-y_n)^2+d^2}}\right)^2.
\end{equation}
To obtain a tractable expression, we apply the high-SINR approximation again here. The ergodic rate can be approximated as \eqref{ergiduc rate multi waveguide} at the top of the next page.
\begin{figure*}[t]
    \begin{align}
   \bar{R}_m^{\rm DF} &\approx \int_0^{L_x} \int_0^{L_y} \log_2 \left(\frac{P_0\eta}{N\sigma^2} e^{-\alpha  \min \left(x, L_x - x \right)}\left(\sum_{n=1}^{N}\frac{1}{\sqrt{(y_m-y_n)^2+d^2}}\right)^2 \right) \frac{1}{L_x} \frac{1}{L_y} \,dx\,dy \\
   & = \log_2 \left(\frac{P_0\eta}{N\sigma^2}\right) - \frac{\alpha}{\ln 2} \frac{1}{L_x}\int_0^{L_x} \min \left(x, L_x - x \right)\,dx + \frac{2}{L_y} \int_0^{L_y} \log_2 \left(\sum_{n=1}^{N}\frac{1}{\sqrt{(y-y_n)^2+d^2}} \right) \, dy \label{ergiduc rate multi waveguide}. 
   \end{align}
\end{figure*}
\begin{lem} \label{lemma 3}
    At high SNR, the ergodic rate of ${\rm U}_m$ with multiple waveguides can be approximated as
    \begin{align}
        \bar{R}^{\rm DF}_m &\approx \log_2 \left(\frac{P_0 \eta}{N\sigma^2} \right) - \frac{\alpha L_x}{4} \log_2(e) \notag \\
         &+2 \log_2 \left(\frac{1}{L_y} \sum_{n=1}^{N} \left[ \operatorname{asinh}\left( \frac{L_y - y_n}{d} \right) + \operatorname{asinh}\left( \frac{y_n}{d} \right) \right] \right). \label{closed form e rate multi waveguide}
    \end{align}
\end{lem}
\begin{proof}
    See Appendix \ref{appendox C}
\end{proof}

\subsection{Optimization}
In this subsection, we aim to maximize the sum rate of all users by jointly optimizing the feed-point selection, the positions of PAs, and the beamforming vectors. 
\subsubsection{Problem Formulation}

Let $\mathbf{P} = \{ \mathbf{p}_1, \dots, \mathbf{p}_M \}$ denote the collection of beamforming vectors. Let $\boldsymbol{\xi} = [\xi_1, \dots, \xi_N]^T$ represent the binary feed selection vector, and $\mathbf{x} = [x_1^{{\rm Pin}}, \dots, x_N^{{\rm Pin}}]^T$ denote the positions of the PAs along the waveguides. The achievable rate for ${\rm U}_m$ is given by $R_m = \log_2(1 + \gamma_m)$. The sum-rate maximization problem is formulated as follows:

\begin{subequations}
\begin{align}
    {\rm P}2: \quad & \underset{\mathbf{P}, \boldsymbol{\xi}, \mathbf{x}}{\max} \quad \sum_{m=1}^M \log_2 \left( 1 + \frac{|\tilde{\mathbf{h}}_m(\boldsymbol{\xi}, \mathbf{x})^H \mathbf{p}_m|^2}{\sum_{i \neq m} |\tilde{\mathbf{h}}_m(\boldsymbol{\xi}, \mathbf{x})^H \mathbf{p}_i|^2 + \sigma^2} \right) \label{prob_obj} \\
    \text{s.t.} \quad & \sum_{m=1}^M ||\mathbf{p}_m||_2^2 \le P_0, \label{prob_power} \\
    & \xi_n \in \{0, 1\}, \quad \forall n = 1, \dots, N, \label{prob_binary} \\
    & 0 \le x_n^{\text{Pin}} \le L_x, \quad \forall n = 1, \dots, N. \label{prob_pos}
\end{align}
\end{subequations}
Constraint \eqref{prob_power} ensures that the total power does not exceed its power budget $P_0$. Constraint \eqref{prob_binary} imposes the binary restriction on the feed-point selection, where $\xi_n=1$ selects the left feed and $\xi_n=0$ selects the right feed. Finally, constraint \eqref{prob_pos} ensures that the PA is physically placed on the waveguide. \par

The formulated problem ${\rm P}1$ represents a mixed-integer non-convex optimization problem due to the binary constraints on $\boldsymbol{\xi}$, the coupling of beamforming vectors in the objective function, and the highly non-linear dependence of the channel vectors $\tilde{\mathbf{h}}_m(\boldsymbol{\xi}, \mathbf{x})$ on the position and switching variables. To efficiently solve this problem, we propose a two-phase optimization framework. In the first phase, the feed-point selection is determined using a greedy switching algorithm. In the second phase, with the feed-point selection fixed, the PA positions and beamforming vectors are alternately optimized using gradient ascent and the WMMSE method, respectively.

\subsubsection{Phase I}
In this phase, we propose a greedy switching algorithm to determine $\boldsymbol{\xi}$. Given fixed beamforming vectors $\mathbf{P}$ and PA positions $\mathbf{x}$, the original problem ${\rm P2}$ reduces to finding the optimal binary feed selection vector $\boldsymbol{\xi}$. This sub-problem is a combinatorial integer programming problem. A brute-force search for the optimal $\boldsymbol{\xi}$ would require evaluating the objective function for all $2^N$ possible combinations, which incurs a prohibitive computational complexity of $\mathcal{O}(2^N)$ and is intractable for large $N$. \par

To address this challenge efficiently, we propose an iterative greedy switching algorithm. Instead of optimizing the entire vector $\boldsymbol{\xi}$ simultaneously, the proposed method sequentially updates the feed point of each waveguide while keeping the others fixed. Specifically, for the $n$-th waveguide, we compare the system sum rate achievable by selecting the left feed ($\xi_n=1$) versus the right feed ($\xi_n=0$). For each case, the beamforming vectors are obtained using the WMMSE method, and the PA position is set as
\begin{equation}
    x_n^{\rm Pin} = \begin{cases}
        0,& \xi_n = 1, \\
        L_x,& \xi_n = 0.
    \end{cases} \label{PA init}
\end{equation}
This process is repeated for all $n=1,\dots,N$ in a cyclic manner until the feed selection vector converges. \par

There are three advantages for adopting this greedy switching strategy over continuous relaxation or exhaustive search. First, regarding computational efficiency, the proposed method significantly reduces the search space from an exponential complexity of $\mathcal{O}(2^N)$ to a linear complexity of $\mathcal{O}(I_{\rm iter} \cdot N)$ per outer iteration, where $I_{\rm iter}$ is the typically small number of greedy passes required. Second, the algorithm guarantees monotonic convergence. Specifically, the feed indicator $\xi_n$ is updated if and only if the change yields a strictly higher sum rate, ensuring the sequence of non-decreasing updates converges to a stationary point. Finally, unlike relaxation methods that approximate $\xi_n$ as a continuous variable in $[0,1]$, this approach maintains physical consistency by strictly respecting the binary nature of the hardware switches, thereby avoiding convergence to non-physical intermediate states (e.g., $\xi_n=0.5$) that would otherwise require heuristic rounding. \par

\subsubsection{Phase II}
In this phase, the PA positions and beamforming vectors are alternately optimized using gradient ascent and the WMMSE method based on the feed-point selection  $\boldsymbol{\xi}^*$ obtained from the greedy switching algorithm. \par 

\textbf{\textit{PA Position Optimization}}: We aim to maximize the sum rate $R_{\rm sum}$ with respect to the PA positions $\mathbf{x}$. Since each element in $\mathbf{x}$ is continuous within the range $[0, L_x]$, we can use gradient ascent method to optimize $\mathbf{x}$. For simplicity, we omit the superscript ${\rm Pin}$ in the following derivation. The update rule for the $n-$th waveguide's PA position at iteration $t$ is given by
\begin{equation}
    x_n^{(t+1)} = \left[x_n^{(t)} + \ell \frac{\partial R_{\rm {sum}}}{\partial x_n}\right]_0^{L_x}, \label{PA positions updating rule}
\end{equation}
where $\ell$ is the step size and $[\cdot]_0^{L_x}$ denotes projection onto the feasible region $[0,L_x]$. \par
To calculate $\frac{\partial R_{\rm sum}}{\partial x_n}$, we need to use the chain rule:
\begin{align}
    \frac{\partial R_{\rm sum}}{\partial x_n} = \sum_{m=1}^{M} \frac{\partial R_m}{\partial \gamma_m} \frac{\partial \gamma_m}{\partial x_n} = \sum_{m=1}^{M} \frac{1}{\ln2 (1+\gamma_m)} \frac{\partial \gamma_m}{\partial x_n}. \label{chain rule}
\end{align}
The SINR for ${\rm U}_m$ can be rewritten as $\gamma_m = \frac{S_m}{I_m+\sigma^2}$, where $S_m = |\tilde{\mathbf{h}}_m^H \mathbf{p}_m|^2$ and $I_m = \sum_{i \neq m} |\tilde{\mathbf{h}}_m^H \mathbf{p}_i|^2$. Since $\tilde{\mathbf{h}}_m$ is complex, we use the property $\frac{\partial |z|^2}{\partial x} = 2\text{Re}\{z^* \frac{\partial z}{\partial x}\}$. Then, we have
\begin{equation}
    \frac{\partial \gamma_m}{\partial x_n} = \frac{1}{I_m + \sigma^2} \frac{\partial S_m}{\partial x_n} - \frac{S_m}{(I_m + \sigma^2)^2} \frac{\partial I_m}{\partial x_n}, \label{SINR gradient}
\end{equation}
where
\begin{equation}
    \frac{\partial S_m}{\partial x_n} = 2\text{Re} \left\{ (\tilde{\mathbf{h}}_m^H \mathbf{p}_m)^* \cdot \left( p^*_{n,m} \frac{\partial \tilde{h}_{n,m}}{\partial x_n} \right) \right\} \label{Sm gradient}
\end{equation}
and
\begin{equation}
    \frac{\partial I_m}{\partial x_n} = \sum_{i \neq m} 2\text{Re} \left\{ (\tilde{\mathbf{h}}_m^H \mathbf{p}_i)^* \cdot \left( p^*_{n,i} \frac{\partial \tilde{h}_{n,m}}{\partial x_n} \right) \right\}. \label{Im gradient}
\end{equation}
The next step is to calculate the gradient $\frac{\partial \tilde{h}_{n,m}}{\partial x_n}$. Note that the effective channel is 
\begin{equation}
    \tilde{h}_{n,m} = \nu_n(x_n) \left( h_{n,m}^{{\rm LoS}}(x_n) + h_{n,m}^{{\rm NLoS}}(x_n) \right).
\end{equation}
Using the product rule, $\frac{\partial \tilde{h}_{n,m}}{\partial x_n}$ can be expressed as
\begin{equation}
    \frac{\partial \tilde{h}_{n,m}}{\partial x_n} = \frac{\partial \nu_n}{\partial x_n} (h^{{\rm LoS}}_{n,m} + h^{{\rm NLoS}}_{n,m}) + \nu_n \left( \frac{\partial h^{{\rm LoS}}_{n,m}}{\partial x_n} + \frac{\partial h^{{\rm NLoS}}_{n,m}}{\partial x_n} \right). \label{channel gradient}
\end{equation}
Recall that $\nu_n = \xi_n e^{-(\frac{\alpha}{2} + j \frac{2\pi}{\lambda_g}) x_n} + (1 - \xi_n) e^{-(\frac{\alpha}{2} + j \frac{2\pi}{\lambda_g})(L_x - x_n)}$, $\frac{\partial \nu_n}{\partial x_n}$ can be expressed as
\begin{equation}
     \frac{\partial \nu_n}{\partial x_n} = \begin{cases}
         -\left(\frac{\alpha}{2} + j \frac{2\pi}{\lambda_g}\right) e^{-(\frac{\alpha}{2} + j \frac{2\pi}{\lambda_g}) x_n}, & \xi_n = 1, \\
          \left(\frac{\alpha}{2} + j \frac{2\pi}{\lambda_g}\right) e^{-(\frac{\alpha}{2} + j \frac{2\pi}{\lambda_g})(L_x - x_n)}, & \xi_n = 0.
     \end{cases}
\end{equation}
To calculate $\frac{\partial h^{{\rm LoS}}_{n,m}}{\partial x_n}$, we use the logarithmic differentiation $f' = f \left(\ln f \right)'$. According to \eqref{LoS channel multi-waveguide}, we have
\begin{equation}
    \frac{\partial \ln h^{\rm LoS}_{n,m}}{\partial x_n}= -j \frac{2 \pi}{\lambda} \frac{\partial r_{n,m}}{\partial x_n} - \frac{1}{r_{n,m}} \frac{\partial r_{n,m}}{\partial x_n}, \label{ln LoS gradient}
\end{equation}
$r_{n,m} = ||\psi^{\mathrm{Pin}}_n - \psi_m|| = \sqrt{(x^{\rm Pin}_n - x_m)^2 + (y_n - y_m)^2 + d^2}$ denotes the Euclidean distance between the PA on the $n-$th waveguide and ${\rm U}_m$. Then, $\frac{\partial h^{{\rm LoS}}_{n,m}}{\partial x_n}$ is given by
\begin{equation}
    \frac{\partial h^{{\rm LoS}}_{n,m}}{\partial x_n} = -h^{\rm LoS}_{n,m} \left(j \frac{2 \pi}{\lambda} + \frac{1}{r_{n,m}} \right) \frac{x_n^{\rm Pin} - x_m}{r_{n,m}}. \label{ LoS gradient}
\end{equation}
We can use the logarithmic differentiation again to calculate $\frac{\partial h^{{\rm NLoS}}_{n,m}}{\partial x_n}$, which is given by
\begin{equation}
    \frac{\partial h_{n,m}^{\rm NLoS}}{\partial x_n} = \sum_{k=1}^{K} \left[ -h_{n,m,k}^{\rm NLoS} \left(j \frac{2 \pi}{\lambda} + \frac{1}{d_{n,k}^{1}} \right) \frac{x_n - x_k^{\rm SC}}{d_{n,k}^{1}} \right], \label{NLoS gradient}
\end{equation}
where $h_{n,m,k}^{\rm NLoS}$ is the $k-$th path component, which is given by
\begin{equation}
    h_{n,m,k}^{\rm NLoS} = \Gamma_k \frac{e^{-j \frac{2\pi}{\lambda} d_{n,k}^{1}(x_n)}}{d_{n,k}^{1}(x_n)} \cdot \frac{e^{-j \frac{2\pi}{\lambda} d_{m,k}^{2}}}{d_{m,k}^{2}}.
\end{equation}
So far, the gradient $ \frac{\partial R_{\rm sum}}{\partial x_n}$ can be calculated by all derived sub-gradients. \par

However, excessively large or small gradients may cause unstable updates. To address this issue, we adopt the normalized gradient to update PA positions. Let $g_n^{(t)} = \frac{\partial R_{\rm sum}}{\partial x_n}$ and $G_{\max}^{(t)} = \max_k \left| \frac{\partial R_{\rm sum}}{\partial x_k} \right|$. The normalized gradient is $\frac{g_n^{(t)}}{G_{\max}^{(t)}}$. Moreover, standard gradient ascent with a fixed step size is prone to oscillation or convergence failure in such non-convex environments. To address this issue, we employ a backtracking line search (BLS) with an adaptive step size strategy. The key idea of BLS is that the step size $\ell$ is not fixed but is dynamically adjusted at each iteration to satisfy the sufficient ascent condition. Let $\ell'^{(t)}$ be the candidate step size at the $t$-th iteration. The actual step size at the $t$-th iteration is given by
\begin{equation}
    \ell^{(t)} = \ell'^{(t)} \cdot \beta^{k^*},
\end{equation}
where $\beta \in (0,1)$ is the contraction factor and $k^*$ is the smallest non-negative integer such that the updated position yields a strictly higher sum rate, satisfying
\begin{equation}
    R_{\rm sum} \left(\left[x_n^{(t)} + \ell'^{(t)}\beta^{k^*} \frac{g_n^{(t)}}{G_{\max}^{(t)}}\right]_0^{L_x}, \forall n\right) > R_{\rm sum} \left(x^{(t)}_n, \forall n\right).
\end{equation}
Once a successful step is found, the candidate step size for the next iteration is expanded to accelerate convergence, which is given by
\begin{equation}
    \ell'^{(t+1)} \leftarrow \zeta \cdot \ell^{(t)},
\end{equation}
where $\zeta > 1$ is the expansion factor. BLS allows the algorithm to aggressively increase the step size when the direction is promising, while quickly backing off when the gradient direction becomes unstable. Finally, the update of the PA position on the $n$-th waveguide at iteration $t$ is then modified as
\begin{equation}
    x_n^{(t+1)} = \left[\, x_n^{(t)} + \ell^{(t)} \frac{g_n^{(t)}}{G_{\max}^{(t)}} \right]_0^{L_x}. \label{modified PA positions updating rule}
\end{equation}
Compared with the original update rule \eqref{PA positions updating rule}, the normalized update in \eqref{modified PA positions updating rule} accelerates the search through adaptive step size, while ensuring robust convergence against the large dynamic range of gradients caused by random user positions and varying SNR conditions.\par

\begin{algorithm}[t]
\caption{Two-phase Optimization Framework}
\begin{algorithmic}[1]
\Require Error tolerance $\epsilon$.
    \State \textbf{Phase I: Greedy Switching for Feed-Point Selection}
    \State \textbf{Initialization:} Initialize $\boldsymbol{\xi}^{(0)}$. Initialize $\mathbf{x}^{(0)}$ by \eqref{PA init}. Initialize $\mathbf{P}^{(0)}$ by \eqref{beamforming update}.  Calculate the sum rate $R_{\rm sum}^{(0)}$.
    \Repeat
        \For{$n = 1$ to $N$}
            \State \textbf{Case 1:} Set $\xi_n^{(t)} = 1$ and temporary $x_n^{{\rm Pin} (t)} = 0$. Solve for $\mathbf{P}^{(t)}$ by \eqref{beamforming update}. Calculate the sum rate $R_{\rm sum}^{l (t)}$.
            \State \textbf{Case 2:} Set $\xi_n^{(t)} = 0$ and temporary $x_n^{{\rm Pin}(t)} = L_x$. Solve for $\mathbf{P}^{(t)}$ by \eqref{beamforming update}. Calculate the sum rate $R_{\rm sum}^{r(t)}$.
            \If{$R_{\rm sum}^{l (t)} \ge R_{\rm sum}^{r (t)}$}
                \State Update $\xi_n^{(t+1)} \leftarrow 1$
            \Else
                \State Update $\xi_n^{(t+1)} \leftarrow 0$
            \EndIf
        \EndFor
        \State Update $\mathbf{x}^{(t+1)}$ by \eqref{PA init} and $\mathbf{P}^{(t+1)}$ by \eqref{beamforming update} according to $\boldsymbol{\xi}^{(t+1)}$. Calculate the sum rate $R_{\rm sum}^{(t+1)}$.
        \State Update $t \leftarrow t + 1$.
    \Until{$R_{\rm sum}^{(t+1)} - R_{\rm sum}^{(t)} < \epsilon$}.
    \State \Return $\boldsymbol{\xi}^* \leftarrow \boldsymbol{\xi}^{(t+1)}$.
        
    \State \textbf{Phase II: Alternating Optimization}
     \State \textbf{Initialization:} Initialize $\mathbf{x}^{(0)}$, $\mathbf{P}^{(0)}$. Calculate the sum rate $R_{\rm sum}^{(0)}$.
    \Repeat

        \State \textbf{Step 1: PA Position Update}
        \State Calculate gradients $\frac{\partial R_{sum}}{\partial x_{n}}$ by \eqref{chain rule}-\eqref{NLoS gradient}.
        \State Update positions $\mathbf{x}^{(t+1)}$ by \eqref{modified PA positions updating rule}.
        
        \State \textbf{Step 2: WMMSE Beamforming}
        \Repeat
            \State Update receiver coefficients $u_m^*$ by \eqref{wmmse receiver co}.
            \State Update MSE weights $w_m^*$ by \eqref{wmmse weight}.
            \State Update beamforming vectors $\mathbf{p}_m^*$ by \eqref{beamforming update}.
        \Until{WMMSE converges}
    \State Update $\mathbf{P}^{(t+1)}$ by $\mathbf{p}_m^*, \forall m$.
    \State Calculate the sum rate $R_{\rm sum}^{(t+1)}$.
    \State Update $t \leftarrow t + 1$.
    \Until{$R_{\rm sum}^{(t+1)} - R_{\rm sum}^{(t)} < \epsilon$}.
\State $\mathbf{x}^* \leftarrow \mathbf{x}^{(t+1)}$, $\mathbf{P}^* \leftarrow \mathbf{P}^{(t+1)}$.
\end{algorithmic} \label{alg two-phase}
\end{algorithm}

\textbf{\textit{Beamforming Design}}: With fixed PA positions $\mathbf{x}$ and feed-point selection $\boldsymbol{\xi}^*$, the channel vectors $\tilde{\mathbf{h}}_m(\boldsymbol{\xi}^*, \mathbf{x}), \forall m$ become constant. Let $\mathbf{h}_m$ denote this fixed effective channel for ${\rm U}_m$. The original problem ${\rm P2}$ reduces to maximizing the sum rate with respect to the beamforming vectors $\mathbf{P}$ under the total power constraint \eqref{prob_power}. To address this, we employ the WMMSE algorithm, which establishes an equivalence between sum rate maximization and the minimization of the WMSE. We introduce two sets of auxiliary variables: the receiver decoding coefficient $u_m \in \mathbb{C}$ and the MSE weight $w_m > 0$ for each ${\rm U}_m$. According to WMMSE algorithm, $e_m$ is expressed as
\begin{equation}
    e_m = |u_m|^2 \left( \sum_{i=1}^M |\mathbf{h}_m^H \mathbf{p}_i|^2 + \sigma^2 \right) - 2 \text{Re}\{u_m^H \mathbf{h}_m^H \mathbf{p}_m\} + 1. \label{wmmse em}
\end{equation}
The problem is then transformed into minimizing the weighted sum-MSE function $\sum_{m=1}^M ( w_m e_m - \ln(w_m) )$ subject to the power constraint.

This equivalent WMMSE problem is convex with respect to each variable block when the others are fixed. Then, the optimal receiver $u_m$ is given by
\begin{equation}
    u_m^* = \left( \sum_{i=1}^M |\mathbf{h}_m^H \mathbf{p}_i|^2 + \sigma^2 \right)^{-1} \mathbf{h}_m^H \mathbf{p}_m. \label{wmmse receiver co}
\end{equation}
Subsequently, the optimal weight $w_m$ is given by
\begin{equation}
    w_m^* =  \frac{1}{e_m(u_m^*)}. \label{wmmse weight}
\end{equation}
Finally, the objective function becomes quadratic in $\mathbf{p}_m$. The optimal beamforming vector is obtained by solving the Karush-Kuhn-Tucker (KKT) conditions, yielding the solution
\begin{equation}
    \mathbf{p}_m^* = \left( \sum_{k=1}^M w_k |u_k|^2 \mathbf{h}_k \mathbf{h}_k^H + \mu \mathbf{I} \right)^{-1} w_m u_m \mathbf{h}_m, \label{beamforming update}
\end{equation}
where $\mu \ge 0$ is the Lagrange multiplier associated with the power constraint. The optimal $\mu$ is found via a one-dimensional bisection search such that $\sum_{m=1}^M ||\mathbf{p}_m(\mu)||^2 = P_0$.

\subsubsection{Overall Algorithm}
To address ${\rm P2}$, we develop a two-phase optimization framework that decomposes the coupled design of feed-point selection, PA placement, and beamforming. In Phase I, a greedy switching strategy is employed to determine the feed-point configuration. For efficiency, PA positions are not continuously optimized in this phase; instead, each PA is temporarily placed at its selected feed point and the beamforming vectors are updated via the WMMSE algorithm. After the feed-point selection converges, Phase II alternates between PA placement and beamforming optimization. Specifically, PA positions are updated by gradient ascent based on the derived rate gradients, while the beamforming vectors are optimized using the WMMSE algorithm under the updated PA positions. This decomposition significantly reduces the computational burden while providing an effective solution to ${\rm P2}$. The details are summarized in Algorithm \ref{alg two-phase}.

\section{Simulation Results}
In this section, we present simulation results for both single-waveguide and multi-waveguide scenarios. For each scenario, the performance of the DF-PAS is evaluated and compared with several benchmark schemes from both analytical and optimization perspectives. The simulation parameters are set as follows: the waveguide height is $d = 1.5$ m, the carrier frequency is $f_c = 28$ GHz, the in-waveguide attenuation coefficient is $\alpha = 1.48$ dB/m, and the noise power spectral density is $\sigma^2 = -90$ dBm/Hz. The scattering environment is modeled with $N_{n,m} = 10$ scattering paths between each PA and each user. All scatterers are assumed to have identical reflection coefficients $\Gamma_k = 0.5, \forall k$. The contraction and expansion factors are set to $\beta = 0.5$ and $\zeta = 1.2$, respectively. \par
\begin{figure*}[t!]
    \centering
    % --- First Subfigure ---
    \begin{subfigure}[b]{0.32\textwidth}
        \centering
        \includegraphics[width=\textwidth]{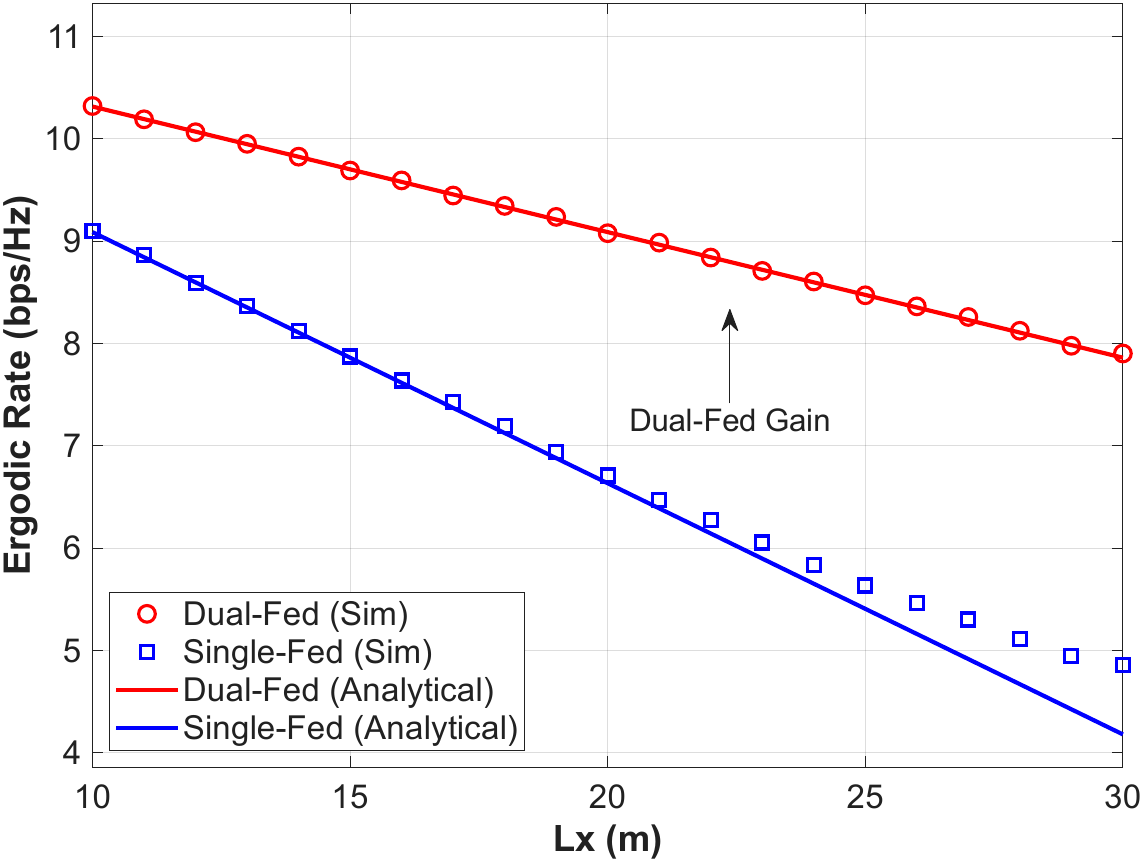}
        \caption{$P_0 = 30$ dBm, $L_y=10$ m}
        \label{single_erate_verify}
    \end{subfigure}
    \hfill
    % --- Second Subfigure ---
    \begin{subfigure}[b]{0.32\textwidth}
        \centering
        \includegraphics[width=\textwidth]{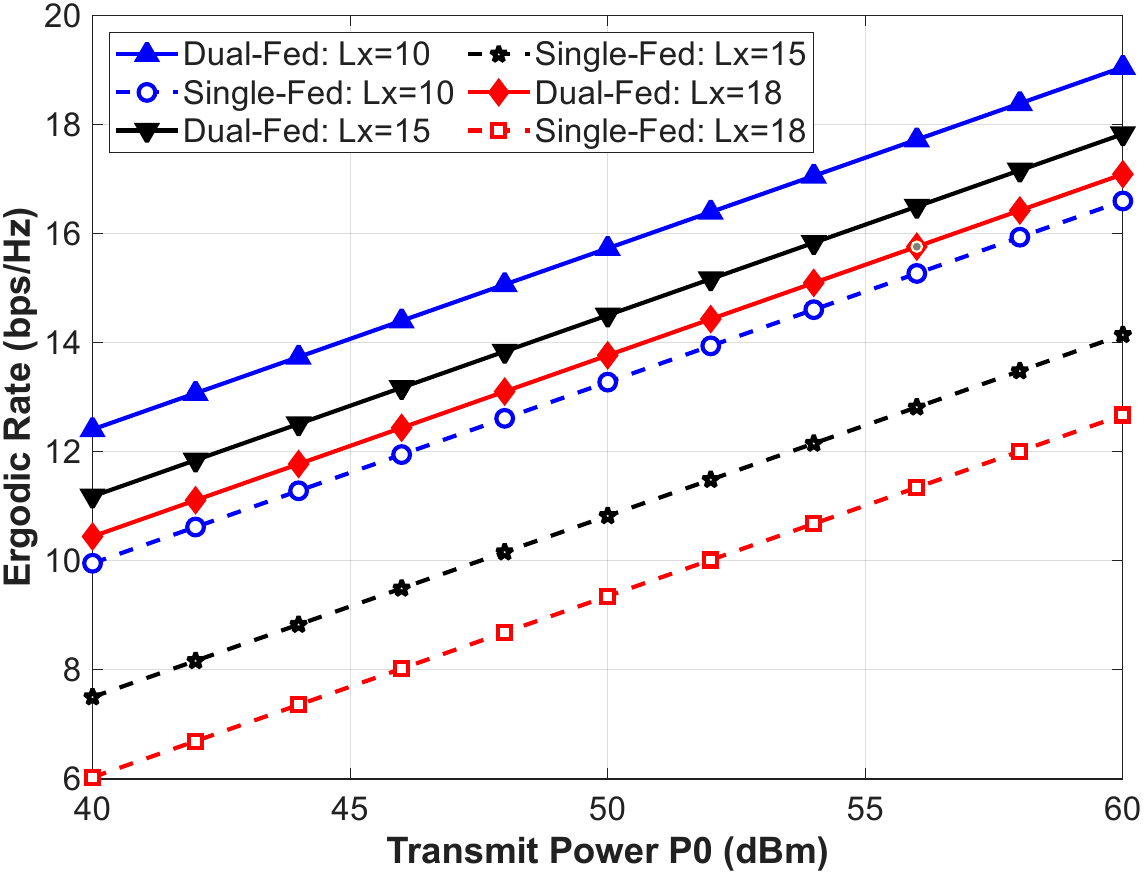}
        \caption{$L_y= 10$ m}
        \label{single_erate_power_Ly}
    \end{subfigure}
    \hfill 
    % --- Third Subfigure ---
    \begin{subfigure}[b]{0.32\textwidth}
        \centering
        \includegraphics[width=\textwidth]{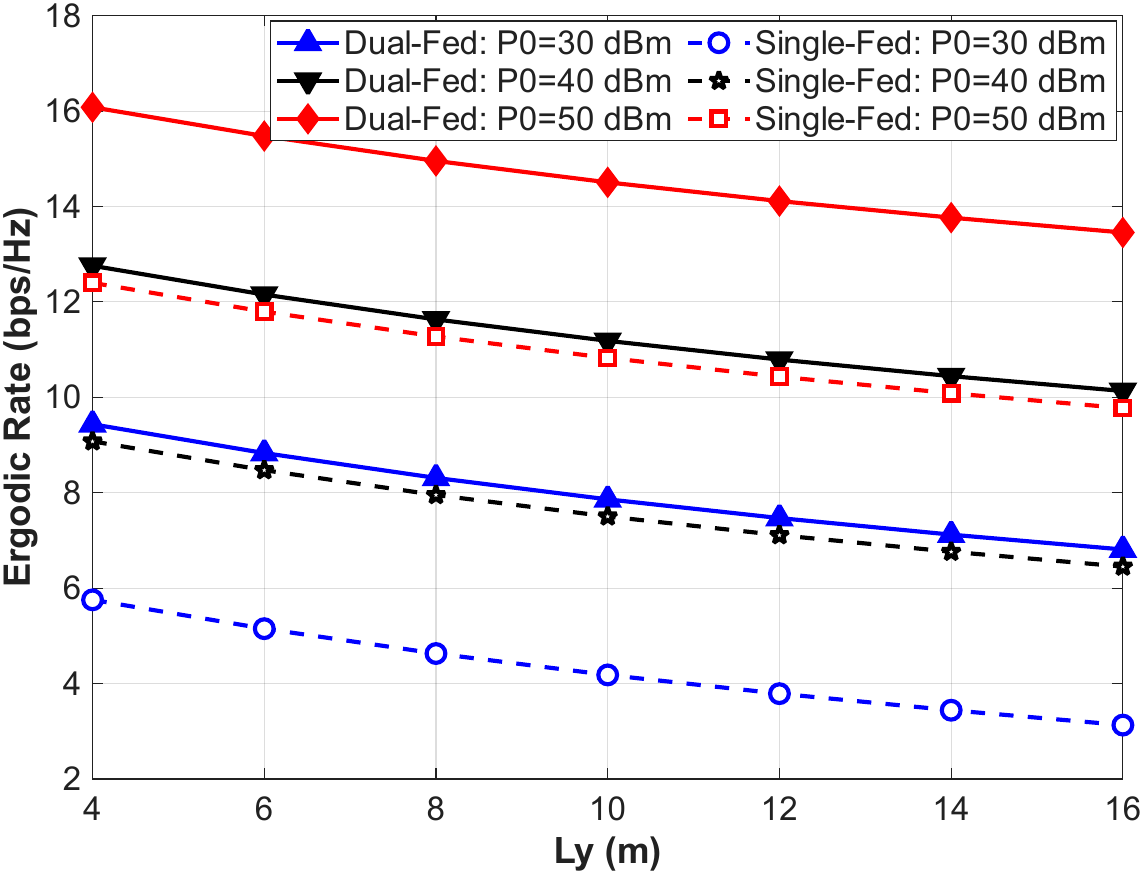}
        \caption{$L_x=15$ m}
        \label{single_erate_Ly_power}
    \end{subfigure}
    \caption{Ergodic rate performance analysis of the single-waveguide scenario: (a) Ergodic data rate versus waveguide length $L_x$, which verifies the accuracy of the derived analytical result \eqref{closed form e rate} against Monte Carlo simulations; (b) Ergodic data rate versus transmit power $P_0$ under different values of $L_x$; (c) Ergodic data rate versus service area width $L_y$ under different transmit power levels $P_0$.}
    \label{single_erate}
\end{figure*}

\begin{figure*}[t!]
    \centering
    % --- First Subfigure ---
    \begin{subfigure}[b]{0.32\textwidth}
        \centering
        \includegraphics[width=\textwidth]{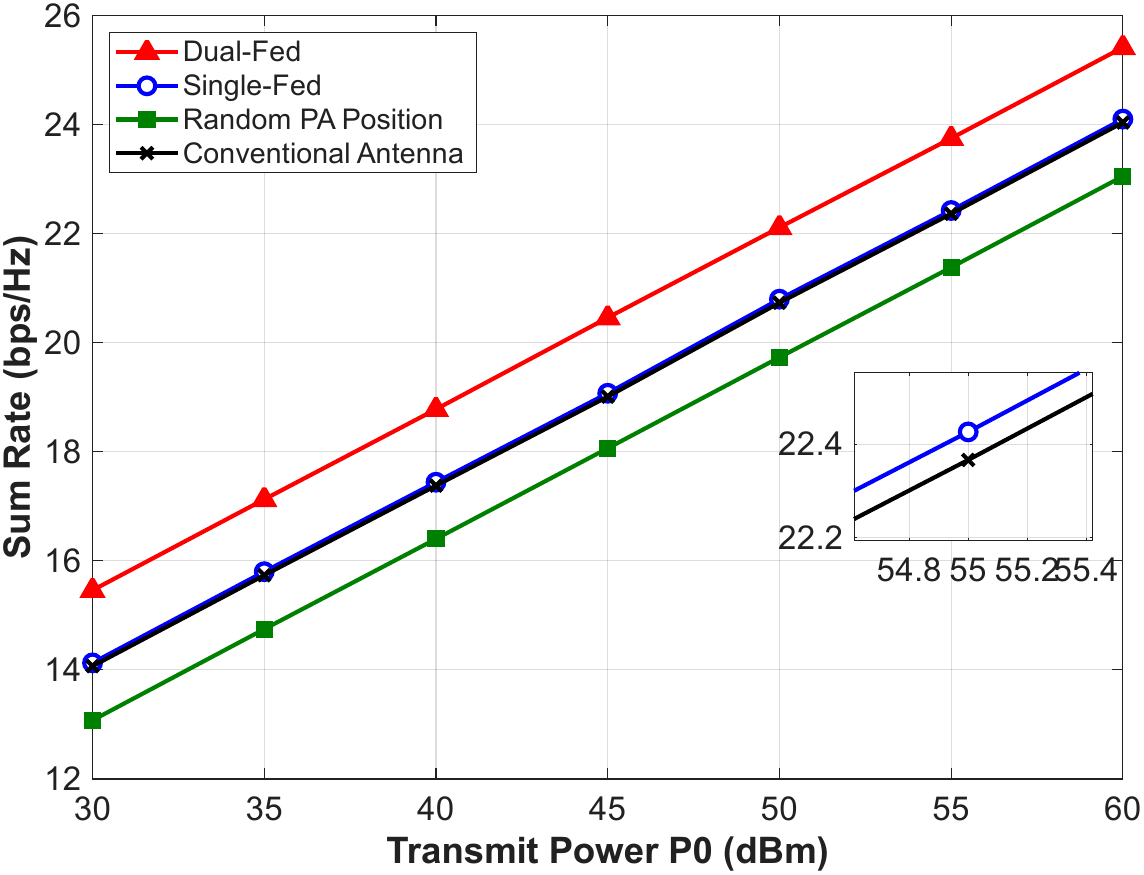}
        \caption{$L_x = 15$ m, $L_y=5$ m}
        \label{op_single_power}
    \end{subfigure}
    \hfill
    % --- Second Subfigure ---
    \begin{subfigure}[b]{0.32\textwidth}
        \centering
        \includegraphics[width=\textwidth]{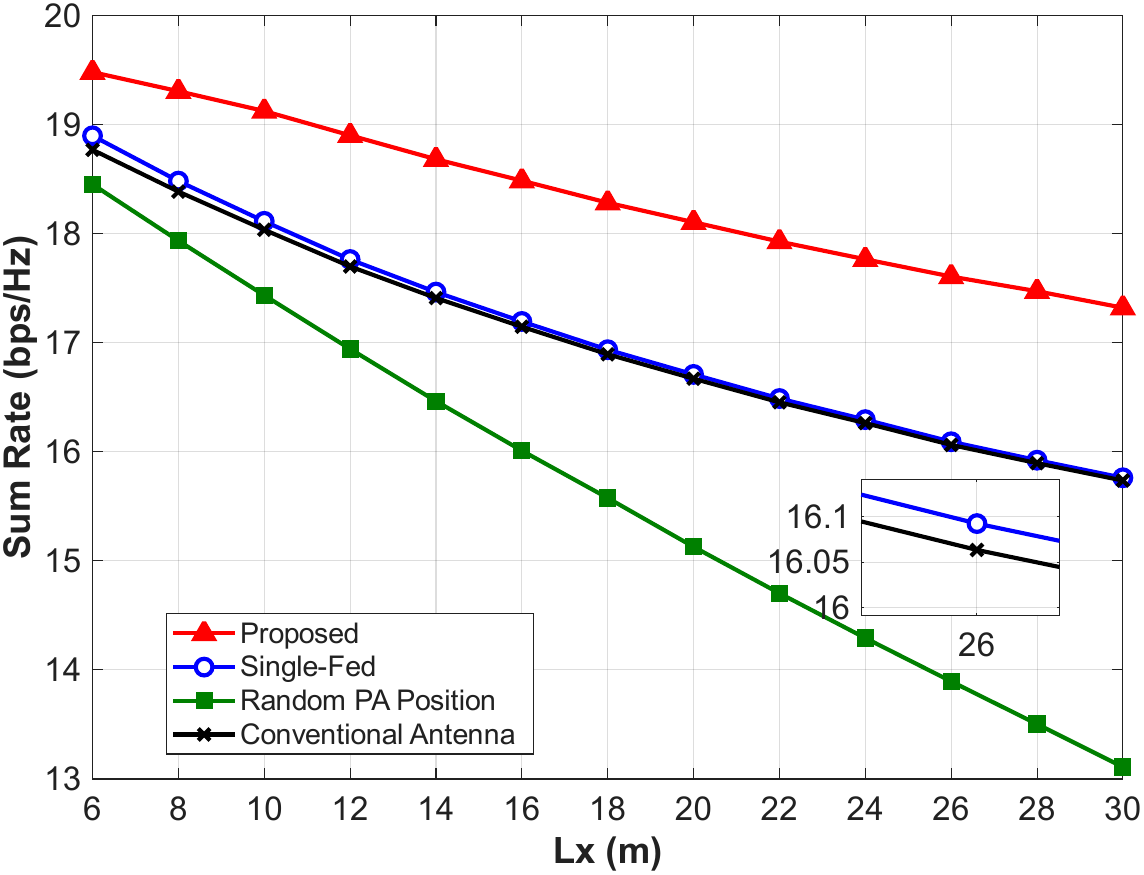}
        \caption{$P_0 = 40$ dBm, $L_y=6$ m}
        \label{op_single_Lx}
    \end{subfigure}
    \hfill 
    % --- Third Subfigure ---
    \begin{subfigure}[b]{0.32\textwidth}
        \centering
        \includegraphics[width=\textwidth]{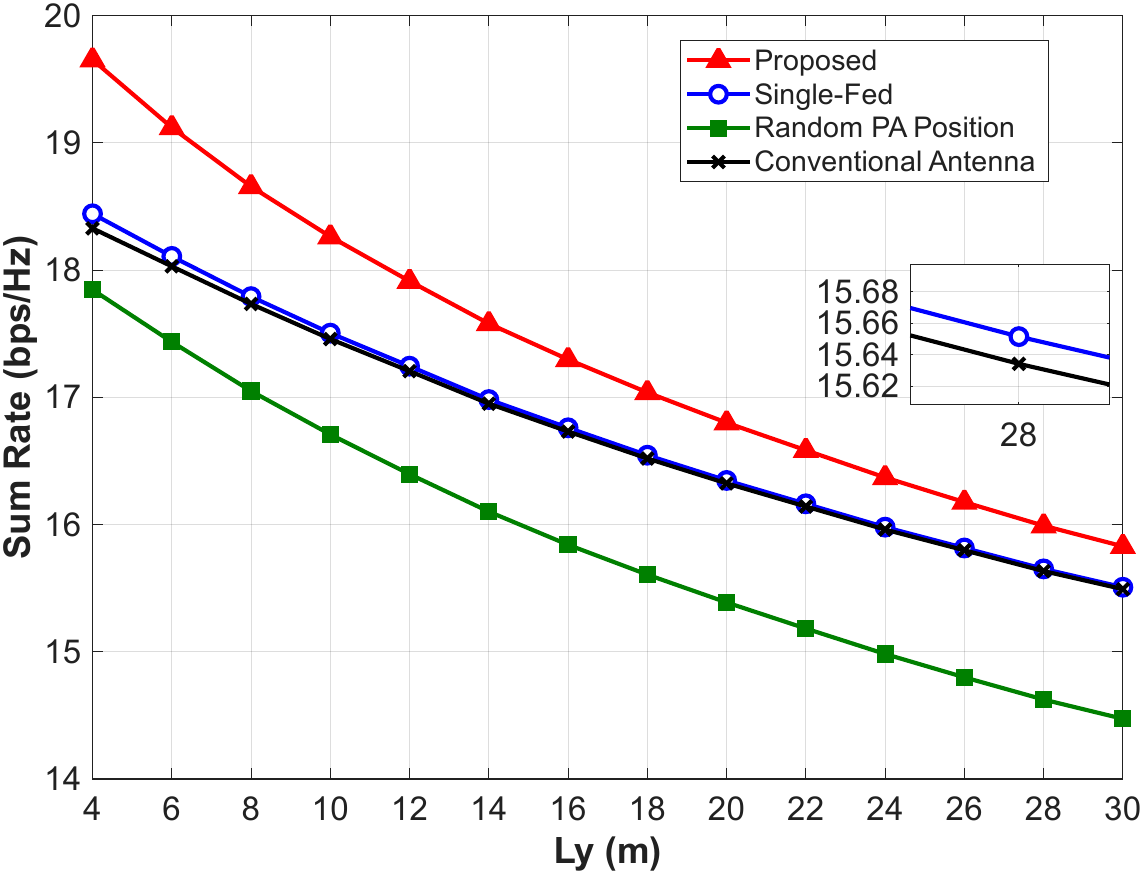}
        \caption{$P_0 = 40$ dBm, $L_x=10$ m}
        \label{op_single_Ly}
    \end{subfigure}
    \caption{Sum rate optimization of the single-waveguide scenario under TDMA: (a) Sum rate versus transmit power $P_0$; (b) Sum rate versus waveguide length $L_x$; (c) Sum rate versus service area width $L_y$.}
    \label{single_op}
\end{figure*}
Fig. \ref{single_erate} illustrates the ergodic rate performance of the single-waveguide DF-PAS and compares it with the corresponding SF-PAS under different system parameters.
\begin{itemize}
    \item In Fig.~\ref{single_erate_verify}, the ergodic rate is plotted versus the waveguide length $L_x$ with $P_0 = 30$~dBm and $L_y = 10$~m. Both Monte Carlo simulation results and the analytical expression in \eqref{closed form e rate} are presented. As $L_x$ increases, the ergodic rate decreases due to the accumulated in-waveguide attenuation. However, the proposed DF-PAS consistently outperforms the SF-PAS and the performance gap between the two schemes widens as $L_x$ increases. Moreover, the close agreement between analytical and simulation curves validates the accuracy of the derived high-SNR approximation.
    \item Fig.~\ref{single_erate_power_Ly} shows the ergodic rate versus the transmit power $P_0$ for different waveguide lengths with $L_y = 10$~m. As expected, increasing $P_0$ leads to a monotonic improvement in the ergodic rate for all schemes. It is also observed that the performance gap between the DF-PAS and SF-PAS becomes larger for larger $L_x$, indicating that the proposed DF-PAS is particularly effective in mitigating severe in-waveguide attenuation in long waveguides.
    \item Fig.~\ref{single_erate_Ly_power} shows the ergodic rate versus the service area width $L_y$ for different transmit power levels with $L_x = 15$~m. As $L_y$ increases, the ergodic rate decreases for both schemes because the average PA-user distance grows. However, the DF-PAS maintains a clear performance advantage across all transmit power levels.
\end{itemize} \par

Fig.~\ref{single_op} illustrates the optimized sum rate performance of the single-waveguide DF-PAS under TDMA and compares it with the SF-PAS, random PA placement, and the conventional antenna benchmark.

\begin{itemize}
    \item Fig.~\ref{op_single_power} shows the sum rate versus the transmit power $P_0$ with $L_x = 15$~m and $L_y = 5$~m. As $P_0$ increases, the sum rate of all schemes improves due to enhanced received signal power. The proposed DF-PAS consistently achieves the highest sum rate over the entire transmit power range. 

    \item Fig.~\ref{op_single_Lx} shows the sum rate versus the waveguide length $L_x$ with $P_0 = 40$~dBm and $L_y = 6$~m. As $L_x$ increases, the sum rate decreases for all schemes due to the accumulation of in-waveguide attenuation. However, the proposed DF-PAS exhibits a significantly slower performance degradation and consistently outperforms the benchmark schemes, demonstrating its robustness against waveguide-length-induced attenuation.

    \item Fig.~\ref{op_single_Ly} shows the sum rate versus the service area width $L_y$ with $P_0 = 40$~dBm and $L_x = 10$~m. As $L_y$ increases, the average PA--user distance grows, which leads to a gradual reduction in the sum rate for all schemes. However, the DF-PAS maintains a clear performance advantage across the entire range of $L_y$, highlighting its robustness against in-waveguide attenuation and free-space path loss.
\end{itemize}

\begin{figure*}[t!]
    \centering
    % --- First Subfigure ---
    \begin{subfigure}[b]{0.32\textwidth}
        \centering
        \includegraphics[width=\textwidth]{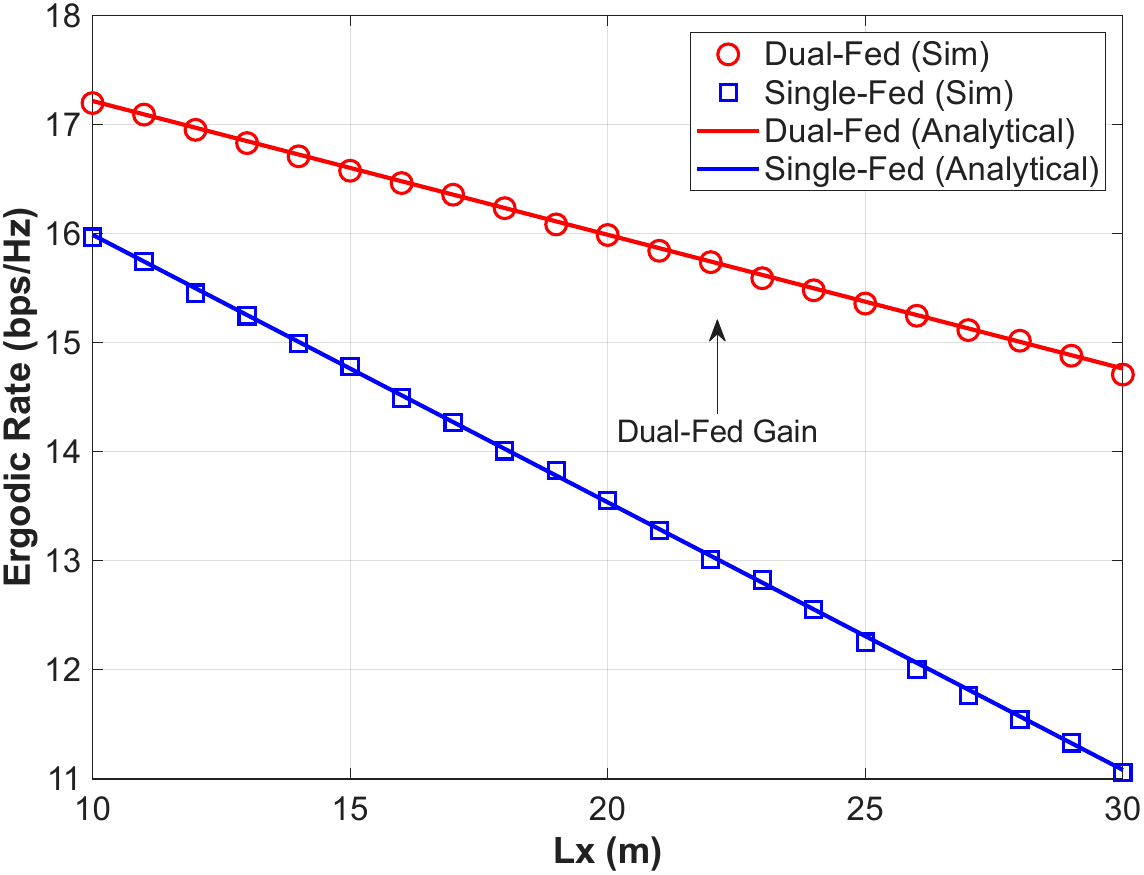}
        \caption{$P_0 = 30$ dBm, $L_y=10$ m}
        \label{multi_erate_verify}
    \end{subfigure}
    \hfill
    % --- Second Subfigure ---
    \begin{subfigure}[b]{0.32\textwidth}
        \centering
        \includegraphics[width=\textwidth]{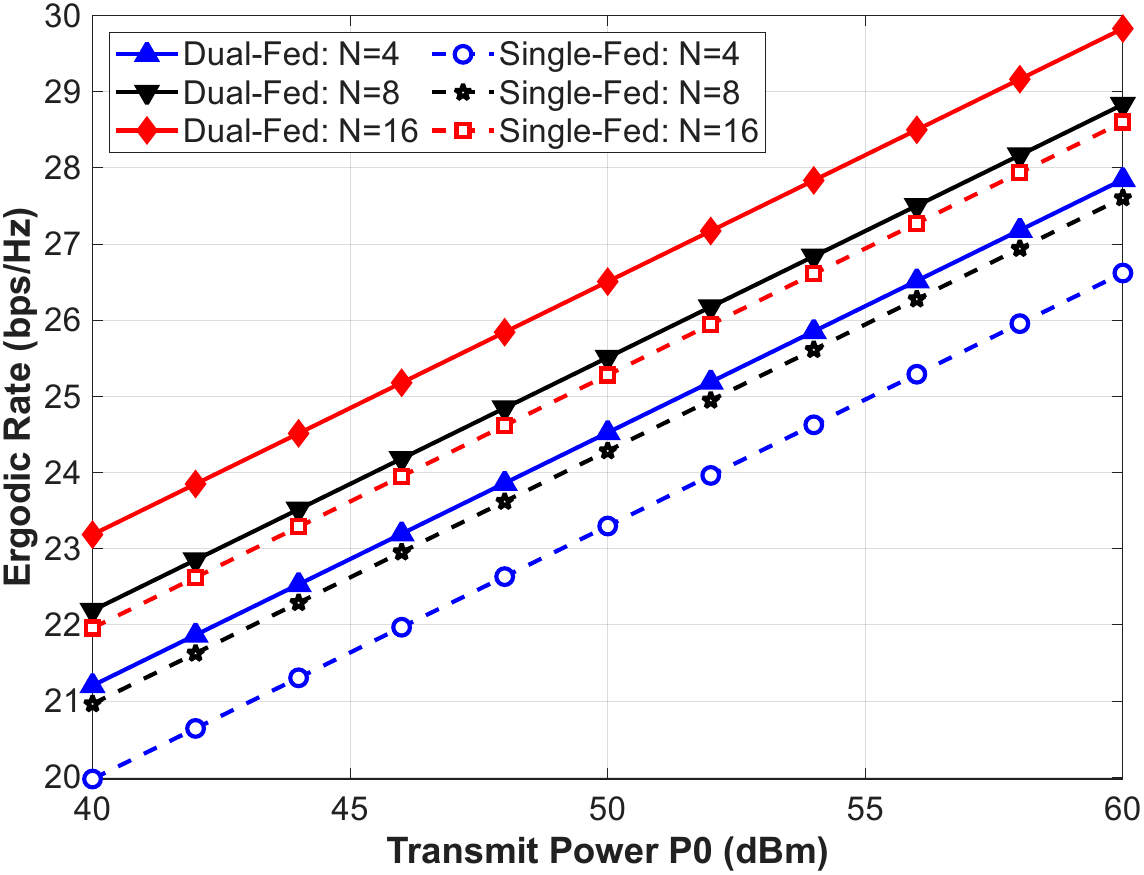}
        \caption{$L_x = 10$ m, $L_y=6$ m}
        \label{multi_erate_power_N}
    \end{subfigure}
    \hfill 
    % --- Third Subfigure ---
    \begin{subfigure}[b]{0.32\textwidth}
        \centering
        \includegraphics[width=\textwidth]{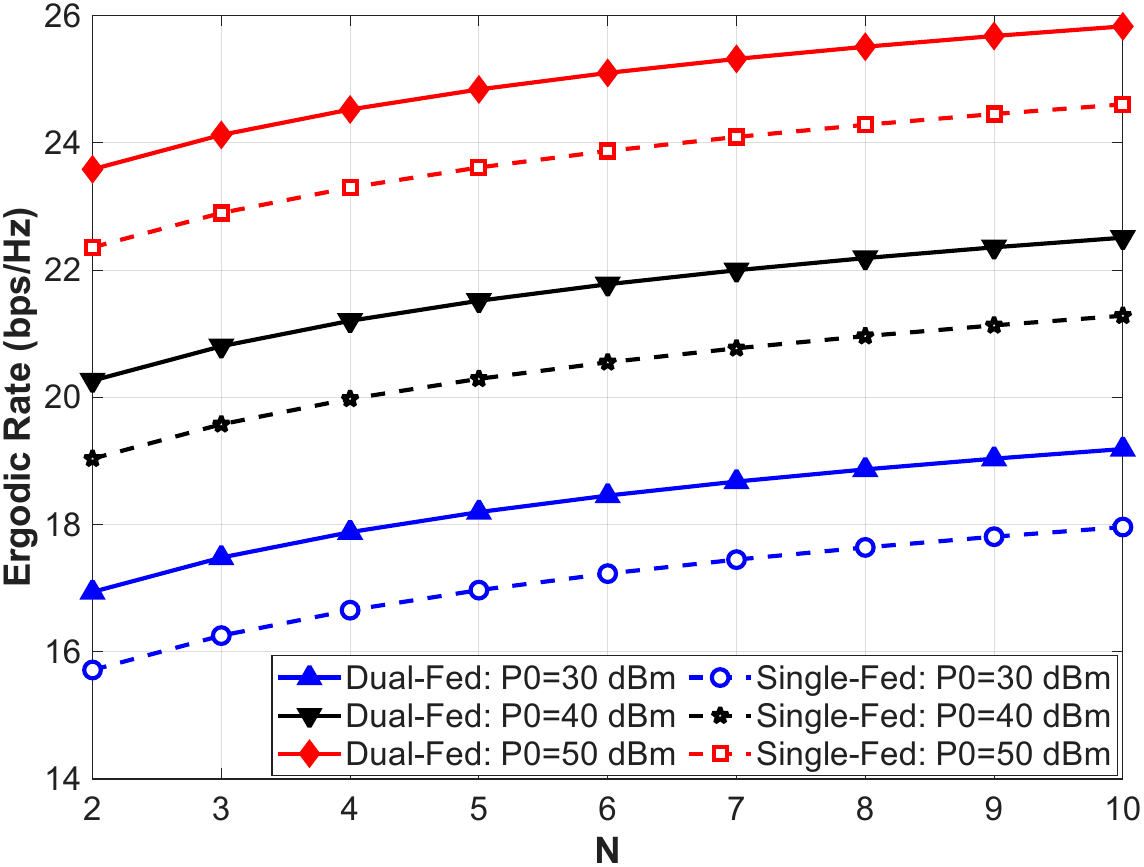}
        \caption{$L_x = 10$ m, $L_y=6$ m}
        \label{multi_erate_N_power}
    \end{subfigure}
    \caption{Ergodic rate performance analysis of the multi-waveguide scenario: (a) Ergodic data rate versus waveguide length $L_x$, which verifies the accuracy of the derived analytical result \eqref{closed form e rate multi waveguide} against Monte Carlo simulations; (b) Ergodic data rate versus transmit power $P_0$ under different numbers of waveguides $N$; (c) Ergodic data rate versus number of waveguides $N$ under different transmit power levels $P_0$.}
    \label{multi_erate}
\end{figure*}

\begin{figure*}[t!]
    \centering
    % --- First Subfigure ---
    \begin{subfigure}[b]{0.32\textwidth}
        \centering
        \includegraphics[width=\textwidth]{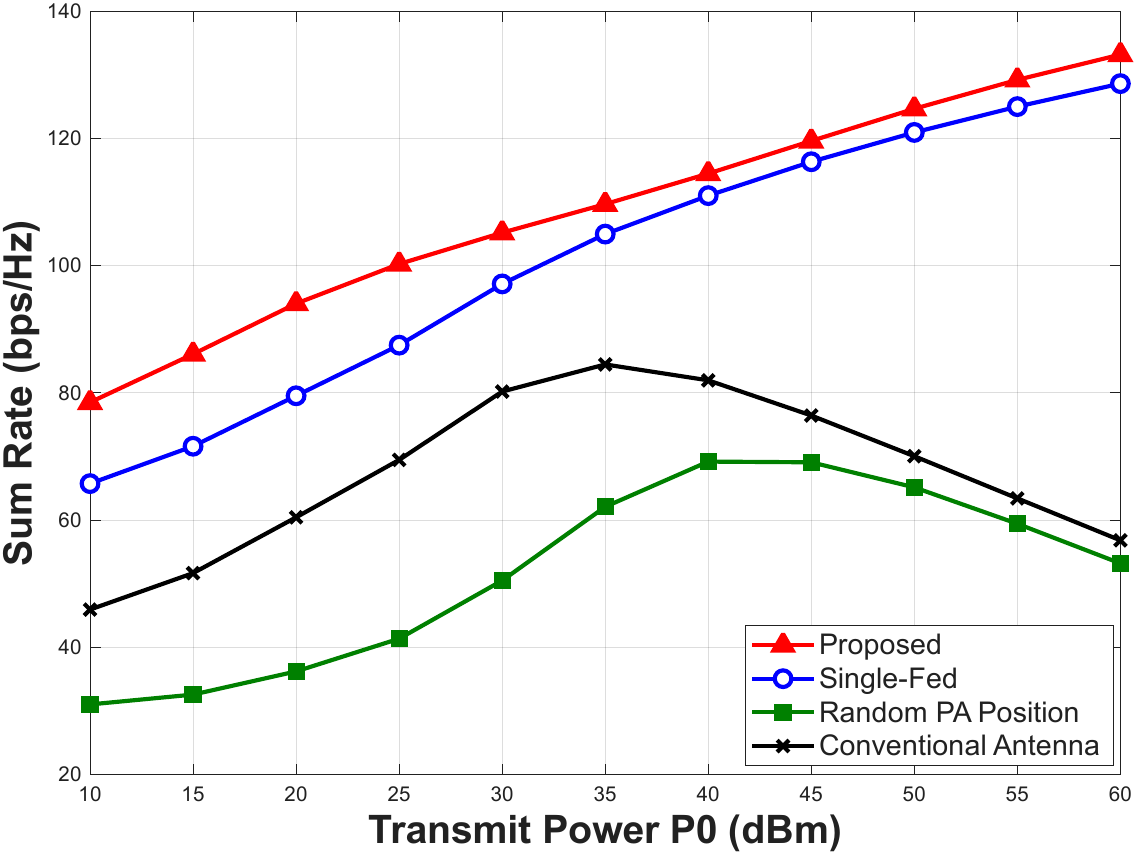}
        \caption{$L_x = 30$ m, $L_y=6$ m}
        \label{op_multi_power}
    \end{subfigure}
    \hfill
    % --- Second Subfigure ---
    \begin{subfigure}[b]{0.32\textwidth}
        \centering
        \includegraphics[width=\textwidth]{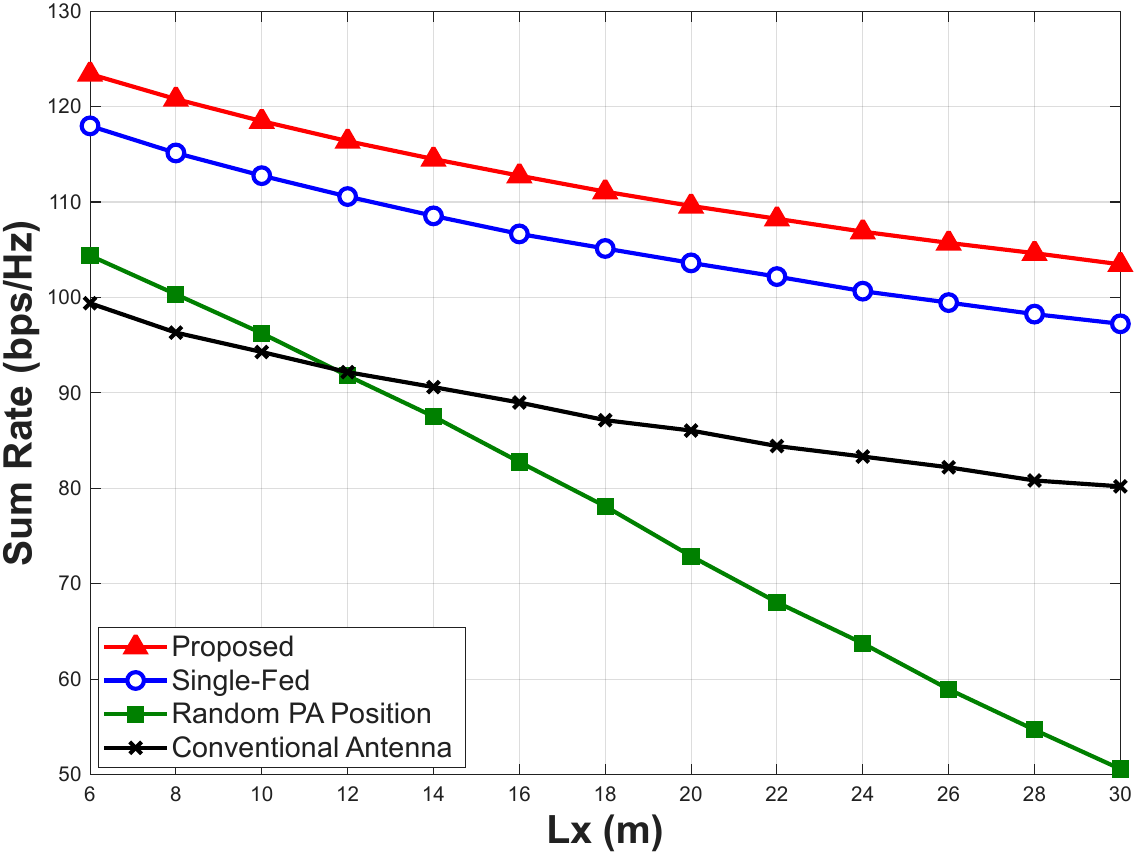}
        \caption{$P_0 = 30$ dBm, $L_y=6$ m}
        \label{op_multi_Lx}
    \end{subfigure}
    \hfill 
    % --- Third Subfigure ---
    \begin{subfigure}[b]{0.32\textwidth}
        \centering
        \includegraphics[width=\textwidth]{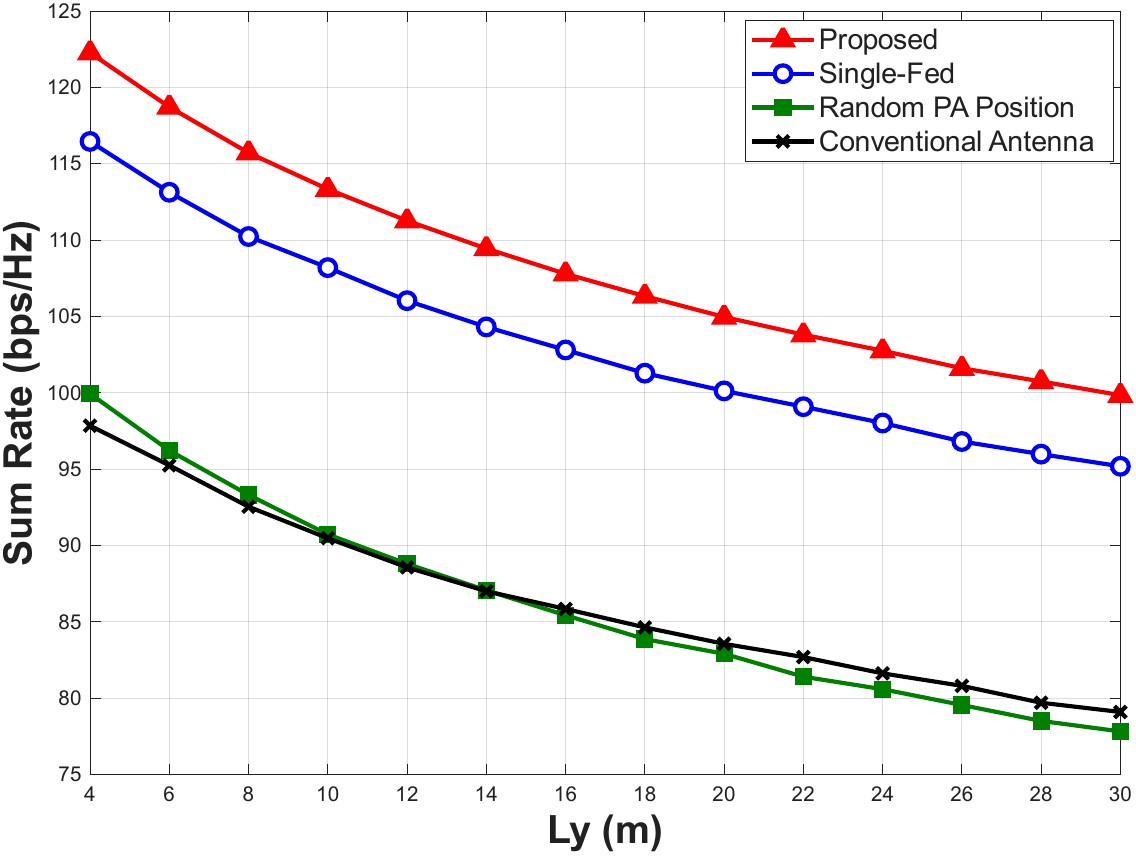}
        \caption{$P_0 = 30$ dBm, $L_x=10$ m}
        \label{op_multi_Ly}
    \end{subfigure}
    \caption{Sum rate optimization of the multi-waveguide scenario under general OMA. (a) Sum rate versus transmit power $P_0$; (b) Sum rate versus waveguide length $L_x$; (c) Sum rate versus service area width $L_y$.}
    \label{multi_op}
\end{figure*}

Fig.~\ref{multi_erate} illustrates the ergodic rate performance of the multi-waveguide DF-PAS and compares it with the corresponding SF-PAS under different system parameters.

\begin{itemize}
    \item Fig.~\ref{multi_erate_verify} shows the ergodic rate versus the waveguide length $L_x$ with $P_0 = 30$~dBm and $L_y = 10$~m. Both Monte Carlo simulation results and the analytical expression in \eqref{closed form e rate multi waveguide} are presented. Similar to the observations in Fig.~\ref{single_erate_verify}, the multi-waveguide DF-PAS consistently outperforms the multi-waveguide SF-PAS across the considered range of $L_x$. Moreover, the close agreement between the analytical and simulation results validates the accuracy of the derived high-SNR approximation.
    
    \item Fig.~\ref{multi_erate_power_N} shows the ergodic rate versus the transmit power $P_0$ for different numbers of waveguides $N$ with $L_x = 10$~m and $L_y = 6$~m. As expected, increasing $P_0$ leads to a monotonic improvement in the ergodic rate for all schemes. In addition, deploying more waveguides significantly enhances the ergodic rate due to increased spatial diversity and array gain. For all values of $N$, the DF-PAS consistently outperforms the SF-PAS, demonstrating the benefit of dual-feed operation in multi-waveguide deployments.

    \item Fig.~\ref{multi_erate_N_power} shows the ergodic rate versus the number of waveguides $N$ under different transmit power levels with $L_x = 10$~m and $L_y = 6$~m. The ergodic rate increases monotonically with $N$ for both schemes, and the performance gain of the DF-PAS remains evident across all transmit power levels. This result confirms that the proposed dual-fed architecture scales well with the number of waveguides while effectively mitigating in-waveguide attenuation.
\end{itemize}

Fig.~\ref{multi_op} illustrates that the optimized sum rate performance of the multi-waveguide DF-PAS under the general OMA transmission scheme and compares it with the multi-waveguide SF-PAS, random PA placement, and the conventional antenna benchmark.

\begin{itemize}
    \item Fig.~\ref{op_multi_power} shows the sum rate versus the transmit power $P_0$ with fixed waveguide length $L_x = 30$~m and service-area width $L_y = 6$~m. As $P_0$ increases, the sum rate of all schemes improves due to enhanced signal power. The proposed DF-PAS consistently achieves the highest sum rate across the entire transmit power range. In contrast, the performance degradation of the random PA position and conventional antenna benchmarks at high transmit power is due to severe multi-user interference. As $P_0$ increases, these schemes become interference-limited because of their lack of spatial adaptability, causing the sum rate to eventually decrease.

    \item Fig.~\ref{op_multi_Lx} shows the sum rate versus the waveguide length $L_x$ with $P_0 = 30$~dBm and $L_y = 6$~m. As $L_x$ increases, the sum rate decreases for all schemes due to the accumulation of in-waveguide attenuation. However, the proposed DF-PAS consistently outperforms all benchmark schemes in the multi-waveguide scenario, demonstrating its robustness against severe waveguide attenuation in large-scale deployments.

    \item Fig.~\ref{op_multi_Ly} shows the sum rate versus the service area width $L_y$ with $P_0 = 30$~dBm and $L_x = 10$~m. Increasing $L_y$ enlarges the average distance between the PAs and users, causing sum rate decrease for all schemes. However, the proposed DF-PAS maintains a clear and stable performance advantage across the entire range of $L_y$ in the multi-waveguide scenario, highlighting its robustness against in-waveguide attenuation and free-space path loss.
\end{itemize}

\section{Conclusion}
This paper proposed a DF-PAS as a novel hardware-efficient architecture to mitigate high in-waveguide attenuation in practical dielectric waveguide based PASs. By structurally enabling feed-point selection from both ends of the waveguide, the DF-PAS reduces the effective in-waveguide propagation distance without modifying the waveguide structure or the PA actuation mechanism. The single-waveguide scenario was investigated under TDMA, while the multi-waveguide scenario was studied under general OMA. For both scenarios, analytical expressions and optimization-based solutions were developed to characterize and improve system performance. Simulation results demonstrate that the DF-PAS consistently outperforms the SF-PAS across a wide range of system parameters. More importantly, the DF-PAS achieves these performance gains through feed-point selection rather than structural redesign or advanced materials. As a result, the DF-PAS provides a hardware-efficient and scalable solution for mitigating in-waveguide attenuation while maintaining low implementation complexity. This makes DF-PASs particularly attractive for practical wireless networks, where deployment cost, coverage robustness, and system flexibility are key considerations.

\begin{appendices}
    \section{Proof for Lemma \ref{lemma 1}} \label{appendox A}
        \eqref{int e rate} can be split into three separate parts by using the logarithm property, which are given by
        \begin{equation}
            I_1 = \frac{1}{L_x L_y} \int_{0}^{L_x} \int_{0}^{L_y} \log_2\left(\frac{P_0 \eta}{\sigma^2}\right)\, dy\, dx, \label{I1}
        \end{equation}
        \begin{equation}
            I_2 = \frac{1}{L_x L_y} \int_{0}^{L_x} \int_{0}^{L_y} \log_2 \left( e^{-\alpha \min(x,\, L_x - x)} \right)\, dy\, dx, \label{I2}
        \end{equation}
        and
        \begin{equation}
            I_3 = \frac{1}{L_x L_y} \int_{0}^{L_x} \int_{0}^{L_y} \log_2\!\left( y^2 + d^2 \right)\, dy\, dx. \label{I3}
        \end{equation}
        The term $I_1$ contains the constant SNR factors, which are not related to $x$ and $y$. The integral of $dx\;dy$ over the area is just the area itself. Therefore, the result of $I_1$ is given by
        \begin{equation}
            I_1 = \log_2 \left(\frac{P_0 \eta}{\sigma^2} \right). \label{result of I1}
        \end{equation}
        The second term $I_2$ accounts for the attenuation inside the waveguide. We note that the integrand depends only on $x$, not $y$. Therefore, $y$ can be integrated out separately. After some algebraic transformations, $I_2$ is recast into
        \begin{equation}
            I_2
            = - \frac{\alpha \log_2(e)}{L_x}\int_{0}^{L_x} \min\!\left(x,\, L_x - x\right)\, dx. \label{I2 step1}
        \end{equation}
        We note that the integral $ \int_{0}^{L_x} \min\!\left(x,\, L_x - x\right)\, dx$ can be split at the midpoint $\frac{L_x}{2}$. Then, the result of $I_2$ can be calculated as follows: 
        \begin{align}
            I_2 &= - \frac{\alpha \log_2(e)}{L_x} \left(\int_{0}^{L_x/2} x\, dx+ \int_{L_x/2}^{L_x} \left( L_x - x \right)\, dx \right) \notag \\
            & = - \frac{\alpha \log_2(e)}{L_x} \cdot \left( \frac{L_x^2}{4} \right) = - \frac{\alpha L_x}{4}\,\log_2(e). \label{result of I2}
        \end{align}
        The third term $I_3$ accounts for the geometric loss from the PA to the user. We note that the integrand depends only on $y$, not $x$. Therefore, $x$ can be integrated out separately. After some algebraic transformations, $I_3$ is recast into
        \begin{equation}
            I_3 = \frac{1}{L_y \ln 2} \int_{0}^{L_y} \ln\!\left( y^2 + d^2 \right)\, dy. \label{I3 step3}
        \end{equation}
        We use the standard integral form $\int \ln\!\left(y^2 + d^2\right)\, dy = y \ln\!\left(y^2 + d^2\right) - 2y + 2d \arctan\!\left(\frac{y}{d}\right)$, the result of $I_3$ can be calculated as follows:
        \begin{align}
            I_3 &= \frac{1}{L_y \ln 2}\left[y \ln\!\left(y^2 + d^2\right)- 2y+ 2d \arctan\!\left(\frac{y}{d}\right)\right]_{0}^{L_y} \notag \\
            &=\frac{\ln\!\left(L_y^2 + d^2\right)}{\ln 2} - \frac{2}{\ln 2} + \frac{2d}{L_y \ln 2} \arctan\!\left(\frac{L_y}{d}\right) \notag \\
            &=\log_2\!\left(L_y^2 + d^2\right) - \frac{2}{\ln 2}\left(1 - \frac{d}{L_y}\arctan\!\left(\frac{L_y}{d}\right)\right). \label{result of I3}
        \end{align}
        \eqref{closed form e rate} can be obtained by combining $I_1 + I_2 - I_3$, which completes the proof of the lemma.

    \section{Proof for Lemma \ref{lemma 2}} \label{appendox B}
    For the SF-PAS, the signal is injected from a fixed end (e.g., $x=0$). The only difference between DF-PASs and SF-PASs is the effective in-waveguide distance. Hence, the SF-PAS has the same $I_1$ and $I_3$ with the DF-PAS. Since $x \sim \mathcal{U}[0, L_x]$, the expected distance is $\mathbb{E}[x] = L_x/2$. Consequently, the ergodic rate of the SF-PAS is given by
    \begin{align}
        \bar{R}^{\rm SF}_m &\approx \log_2 \left(\frac{P_0 \eta}{\sigma^2} \right) - \frac{\alpha L_x}{2} \log_2(e) \notag \\
         &-\left[\log_2 \left(L_y^2 + d^2 \right) - \frac{2}{\ln2} \left(1-\frac{d}{L_y} \arctan \left(\frac{L_y}{d}\right) \right) \right]. \label{closed form e rate of SF}
    \end{align}
    Subtracting $\bar{R}^{\rm SF}_m$ from $\bar{R}^{\rm DF}_m$ yields \eqref{e rate gain}. Since $\alpha, L_x > 0$, the gain $\Delta \bar{R}$ is strictly positive. This completes the proof. 

    \section{Proof for Lemma \ref{lemma 3}} \label{appendox C}
   The second integral in \eqref{ergiduc rate multi waveguide} is $-I_2$ in Appendix \ref{appendox A}. Therefore, the result is given by
   \begin{equation}
       \frac{\alpha}{\ln 2} \frac{1}{L_x}\int_0^{L_x} \min \left(x, L_x - x \right)\,dx = \frac{\alpha L_x}{4}\,\log_2(e). \label{append c1}
   \end{equation}
   We note that the $\log$ function is concave. By applying Jensen's inequality $\mathbb{E}[\log_2 x] \leq \log_2 \mathbb{E}[x]$, we have 
   \begin{align}
       \int_0^{L_y} \log_2 \left(\sum_{n=1}^{N}\frac{1}{\sqrt{(y-y_n)^2+d^2}} \right) \frac{1}{L_y}\, dy \notag \\
       \leq \log_2 \left(\sum_{n=1}^{N}  \int_0^{L_y} \frac{1}{\sqrt{(y-y_n)^2+d^2}} \frac{1}{L_y}\,dy\right). \label{upper bound}
   \end{align}
    For dense waveguide deployments, $\sum_{n=1}^{N}\frac{1}{\sqrt{(y_m-y_n)^2+d^2}}$ concentrates around its mean, making this bound tight. We therefore adopt this upper bound to approximate $\int_0^{L_y} \log_2 \left(\sum_{n=1}^{N}\frac{1}{\sqrt{(y-y_n)^2+d^2}} \right) \frac{1}{L_y}\, dy$.
    By using the anti-derivation $\int \frac{1}{\sqrt{u^2 + d^2}} \, du = \operatorname{asinh}\left(\frac{u}{d}\right)$, where $\operatorname{asinh}(x) = \ln \left(x + \sqrt{x^2 +1} \right)$, we have
    \begin{align}
        &\log_2 \left(\sum_{n=1}^{N}  \int_0^{L_y} \frac{1}{\sqrt{(y-y_n)^2+d^2}} \frac{1}{L_y}\,dy\right) \notag \\
        & =\log_2 \left(\frac{1}{L_y}\sum_{n=1}^{N} \operatorname{asinh} \left(\frac{L_y - y_n}{d}\right) - \operatorname{asinh} \left(\frac{-y_n}{d}\right)\right)  \notag \\
        &\overset{(b)}{=}\log_2 \left(\frac{1}{L_y}\sum_{n=1}^{N} \operatorname{asinh} \left(\frac{L_y - y_n}{d}\right) + \operatorname{asinh} \left(\frac{y_n}{d}\right)\right) \label{append c3}.
    \end{align}
    $(b)$ is because $\operatorname{asinh}$ is an odd function. By replacing the second and third terms in \eqref{ergiduc rate multi waveguide} with \eqref{append c1} and \eqref{append c3}, respectively, this completes the proof.
    
\end{appendices}

\bibliographystyle{IEEEtran}
\bibliography{IEEEfull,pinchingwaveguide}
\end{document}